\documentclass[aps,prx,reprint,superscriptaddress,nofootinbib,longbibliography]{revtex4-2}
\usepackage{xcolor}
\usepackage{amsmath,amssymb,mathtools,bm}
\usepackage{amsthm}
\usepackage{graphicx}
\usepackage{microtype}
\usepackage[hidelinks]{hyperref}

\newcommand{\Hh}{\mathcal H}
\newcommand{\Gg}{\mathcal G}
\newcommand{\Ff}{\mathfrak F}
\newcommand{\Pp}{\mathfrak P}
\newcommand{\Ran}{\operatorname{Ran}}
\newcommand{\Dom}{\operatorname{Dom}}
\newcommand{\Sep}{\operatorname{Sep}}
\newcommand{\Seg}{\operatorname{Seg}}
\newcommand{\Bisep}{\operatorname{Bisep}}
\newcommand{\Tr}{\operatorname{Tr}}
\newcommand{\Part}{\operatorname{Part}}
\newcommand{\id}{\mathbb I}
\newcommand{\clconv}{\overline{\operatorname{conv}}^{\|\cdot\|_1}}
\newcommand{\supp}{\operatorname{supp}}

\newtheorem{theorem}{Theorem}
\newtheorem{lemma}{Lemma}
\newtheorem{proposition}{Proposition}
\newtheorem{corollary}{Corollary}

\begin{document}

\title{Full Inseparability and Genuine Multipartite Entanglement Coincide for Finite-Mode Gaussian States}

\author{Yu Yang}
\affiliation{School of Applied Science, Beijing Information Science and Technology University, Beijing 100192, China}

\author{Yuhang Wang}
\affiliation{School of Instrument Science and Opto-Electronics Engineering, Beijing Information Science and Technology University, Beijing 100192, China}

\author{Chunxiao Du}
\affiliation{School of Physics, Beihang University, Beijing 100191, China}

\author{Shikun Zhang}
\affiliation{School of Future Technology ,Henan University,Zhengzhou 450046,China}

\author{Zheng Qin}
\affiliation{Shenzhen Institute of Beihang University ,Henan University,Shenzhen 518063,China}
\affiliation{Jiangxi Beidouyun Intelligent Technology Co. Ltd.,Nanchang 330038,China}

\author{Rui Li}
\email{rli.work@buaa.edu.cn}
\affiliation{School of Applied Science, Beijing Information Science and Technology University, Beijing 100192, China}

\author{Wenxiu Li}
\affiliation{School of Automation (School of Artificial Intelligence), Beijing Information Science and Technology University, Beijing 100192, China}

\author{Hao Zhang}
\affiliation{School of Space and Earth Sciences, Beihang University, Beijing 100191, China}

\author{Zhisong Xiao}
\affiliation{School of Instrument Science and Opto-Electronics Engineering, Beijing Information Science and Technology University, Beijing 100192, China}
\affiliation{School of Physics, Beihang University, Beijing 100191, China}

\begin{abstract}
For general mixed states, entanglement across every bipartition need not imply genuine multipartite entanglement (GME), because a biseparable decomposition may switch the separable cut from term to term. We prove that this convex ambiguity disappears for Gaussian states of finitely many bosonic modes. More generally, for any finite family of partitions, a Gaussian density operator in the trace-norm-closed convex class generated by states separable across those partitions is already separable across one fixed partition in the family. Only the target is Gaussian; a valid decomposition may be continuous and may contain arbitrary non-Gaussian states. Thus full inseparability and GME coincide, Gaussian $k$-separability and $k$-producibility reduce to fixed-partition tests, and party-wise tensor powers cannot activate GME from a biseparable Gaussian state. The proof combines a spectral selector with a holomorphic rigidity argument that converts one product vector in the square-root range of a Gaussian state into a block-local covariance certificate. The result shows that partition mixing, a generic mixed-state mechanism, adds no new exact finite-mode Gaussian states.
\end{abstract}

\maketitle

\section{Introduction}

For general mixed multipartite states, fixed-partition inseparability and genuine multipartite entanglement (GME) are different notions. A state may be entangled across every bipartition while still being biseparable, because a convex decomposition may use different bipartitions in different terms \cite{DurCiracTarrach1999,Horodecki2009,GuhneToth2009}. The same freedom appears in $k$-separability and $k$-producibility \cite{Gabriel2010,GuhneTothBriegel2005,SorensenMolmer2001,Lu2018}. Thus, for a generic mixed state, convexification over partitions creates states that belong to no single fixed-partition separable set.

Gaussian states are a natural place to ask whether this convexification remains genuinely effective. They form the basic continuous-variable model in quantum optics and are described by finitely many first and second moments \cite{Weedbrook2012,Walschaers2021,Serafini2017}. Fixed-bipartition Gaussian separability has exact covariance characterizations \cite{WernerWolf2001,Giedke2001}; in the special $1\times1\times1$ three-mode setting, Giedke \emph{et al.} also gave the corresponding three-party covariance-domination criterion as Theorem~$2'$ of Ref.~\cite{GiedkeTripartite2001}. The question studied here is different. Gaussianity is imposed only on the \emph{target density operator}, not on a separable decomposition: the components may be arbitrary non-Gaussian states, and in infinite-dimensional Hilbert space a valid decomposition may require a continuous ensemble \cite{HolevoShirokovWerner2005}. The issue is therefore a state-space one: can mixing separability partitions create a Gaussian target that no single partition can produce?

This distinction is not captured by mixed-partition covariance relaxations. Second-moment entanglement criteria and their convex-optimization formulations provide useful practical tests \cite{HyllusEisert2006,ShchukinVanLoock2026}. The importance of combining information from different multipartite partitions has also been demonstrated experimentally in continuous-variable systems \cite{Gerke2016}. However, covariance compatibility need not coincide with state-space biseparability, even when the target itself is Gaussian \cite{Baksova2025}. Minimal second-moment criteria give efficient sufficient tests for GME \cite{Leskovjanova2025}. Most directly, a recent projection-based study found no fully inseparable biseparable states in several representative Gaussian families and conjectured that every fully inseparable finite-mode Gaussian state is genuinely multipartite entangled \cite{Leskovjanova2026}.

We prove a stronger statement. For any finite family of partitions, a Gaussian target in the trace-norm-closed convex class generated by the corresponding fixed-partition separable sets is already separable across one partition in that family. In other words, \emph{convexification over partitions creates no new Gaussian target states}. The bipartition case proves the conjectured equivalence between full inseparability (FIS) and GME. For an exact finite-mode Gaussian target, any valid demonstration of inseparability across every bipartition is therefore already an exact GME certificate; no additional witness against decompositions that mix different cuts is required. The same rigidity removes partition mixing from Gaussian $k$-separability and $k$-producibility and turns the corresponding membership problem into a finite disjunction of fixed-partition covariance feasibility tests. It also rules out party-wise multicopy GME activation from an exact biseparable Gaussian state: one fixed separable cut already exists and survives every tensor power when each party holds its local copies.

Two obstacles separate this statement from the familiar fixed-partition covariance theory. First, a continuous convex decomposition does not select any single partition. We overcome this by a \emph{one-power spectral selector}, which extracts one product vector in a subcritical spectral range of the target. Second, one must turn the existence of a single product vector in a square-root range into a global separability statement for the Gaussian density operator. A \emph{block-holomorphic contact mechanism} does this by converting product structure plus a Gaussian growth envelope into a block-local pure Gaussian covariance below the target covariance. For faithful Gaussian states, normalized positive powers remain inside the Gaussian family \cite{SeshadreesanLamiWilde2018}; below we show directly that the same interpolation also extends to nonfaithful (non-full-support) boundary points. Figure~\ref{fig:proof-roadmap} previews the global argument and the local rigidity mechanism.

The scope is finite and exact: we consider finitely many finite-mode parties and unrestricted state-space separability of exact Gaussian targets, not Gaussian-only decompositions or covariance-only mixed-partition compatibility.

\section{Convex partition rigidity}
\label{sec:mainresult}

\subsection{Definitions and scope}

Consider finitely many physical parties $A_1,\ldots,A_n$, each containing finitely many bosonic modes, with
\begin{equation}
 \Hh=\bigotimes_{j=1}^n\Hh_{A_j}.
\end{equation}
Each local bosonic Fock space is separable, and so is $\Hh$; for a quantum-optics introduction to Fock-space and phase-space descriptions, see Refs.~\cite{Weedbrook2012,Walschaers2021}. Let $[n]=\{1,\ldots,n\}$ and let $\Part([n])$ be the finite set of set partitions of $[n]$. For
\begin{equation}
 \pi=\{B_1,\ldots,B_r\},\qquad
 \Hh_{B_\alpha}=\bigotimes_{j\in B_\alpha}\Hh_{A_j},
\end{equation}
define the pure $\pi$-product projectors
\begin{equation}
 \Pp_\pi
 :=\left\{
 \left|\bigotimes_{\alpha=1}^r\phi_\alpha\right\rangle
 \left\langle\bigotimes_{\alpha=1}^r\phi_\alpha\right|:
 \|\phi_\alpha\|=1
 \right\}
 \label{eq:pure-pi-product}
\end{equation}
and the fixed-partition separable set \cite{Werner1989,HolevoShirokovWerner2005}
\begin{equation}
 \Sep_\pi:=\clconv(\Pp_\pi).
 \label{eq:fixed-partition-sep}
\end{equation}
The trace-norm closure is part of the definition; in infinite-dimensional state spaces, separable states need not admit countable pure-product decompositions \cite{HolevoShirokovWerner2005}. If $\pi$ refines $\sigma$, then
\begin{equation}
 \Sep_\pi\subseteq\Sep_\sigma.
 \label{eq:refinement-inclusion}
\end{equation}
For a nonempty family $\Ff\subseteq\Part([n])$, define
\begin{equation}
 \Sep_{\Ff}
 :=\clconv\!\left(\bigcup_{\pi\in\Ff}\Sep_\pi\right).
 \label{eq:family-sep}
\end{equation}
Because $n$ is finite, every such family $\Ff$ is finite.

We use the following continuous-variable quadrature convention, equivalent up to the usual rescaling of quadratures to conventions commonly used in Gaussian quantum information \cite{Weedbrook2012,Walschaers2021}
\begin{equation}
 R=(q_1,p_1,\ldots,q_m,p_m)^T,\qquad [R_j,R_k]=i\Omega_{jk},
\end{equation}
with
\begin{equation}
 \Omega=\bigoplus_{j=1}^m\begin{pmatrix}0&1\\-1&0\end{pmatrix},
\end{equation}
and covariance
\begin{equation}
 V_{jk}=\langle\Delta R_j\Delta R_k+\Delta R_k\Delta R_j\rangle,
 \qquad V_{\rm vac}=\id.
 \label{eq:covariance-convention}
\end{equation}
Thus physicality is $V+i\Omega\succeq0$ \cite{Robertson1929,Schrodinger1930,Weedbrook2012,Serafini2017}. A phase-space (Weyl) displacement $D(d)$ \cite{Weedbrook2012,Walschaers2021,SeshadreesanLamiWilde2018} factorizes over every partition,
\begin{equation}
 D(d)=\bigotimes_{\alpha=1}^rD_{B_\alpha}(d_\alpha).
 \label{eq:partition-local-displacement}
\end{equation}

\paragraph{Scope and sensitive hypotheses.}
Five scope conditions should be kept explicit from the outset. The number of physical parties and the number of modes per party are finite, so every allowed partition family is finite. Only the target density operator is assumed Gaussian; decomposition components may be arbitrary and non-Gaussian. Separability means the trace-norm-closed state-space notion and may involve uncountable continuous ensembles. Gaussian targets need not be faithful, meaning that the density operator need not have full support; for Gaussian states in the convention $V_{\rm vac}=\id$, faithfulness is equivalent to all symplectic eigenvalues being strictly larger than one \cite{SeshadreesanLamiWilde2018}. Finally, the statement concerns exact finite-mode states and must not be identified with covariance-only mixed-partition compatibility or with an approximate-Gaussian robustness claim. Additional conventions needed only for Gaussian powers, first moments, hierarchy measures, and tensor powers are stated where they enter.

\subsection{Main theorem}

\begin{theorem}[Convex partition rigidity]
\label{thm:main}
Let $\Gg$ be the set of Gaussian states on a fixed finite number of bosonic modes distributed among finitely many parties. For every nonempty $\Ff\subseteq\Part([n])$,
\begin{equation}
 \boxed{
 \Gg\cap\Sep_{\Ff}
 =
 \Gg\cap\bigcup_{\pi\in\Ff}\Sep_\pi .}
 \label{eq:main-theorem}
\end{equation}
\end{theorem}

In words, convex mixing among different separability partitions creates no new Gaussian states. For general mixed states the closed convex hull on the left can be strictly larger than the union on the right; the theorem says that the added region contains no Gaussian target state.

Let $\Ff_{\rm bip}$ be the finite family of all bipartitions of the physical parties and define the state-space biseparable class $\Bisep:=\Sep_{\Ff_{\rm bip}}$. Since a state is GME if and only if it is not in $\Bisep$, while it is fully inseparable if and only if it belongs to no fixed-bipartition separable class, Theorem~\ref{thm:main} immediately gives the title result.

\begin{corollary}[Full inseparability equals genuine multipartite entanglement]
\label{cor:fis-gme}
For every finite-mode Gaussian state,
\begin{equation}
 \boxed{
 \rho_G\text{ is fully inseparable}
 \quad\Longleftrightarrow\quad
 \rho_G\text{ is GME}.}
\end{equation}
\end{corollary}

Figure~\ref{fig:convex-partition-rigidity} gives the geometric picture behind the theorem.

\begin{figure*}[t]
  \centering
  \includegraphics[width=\textwidth]{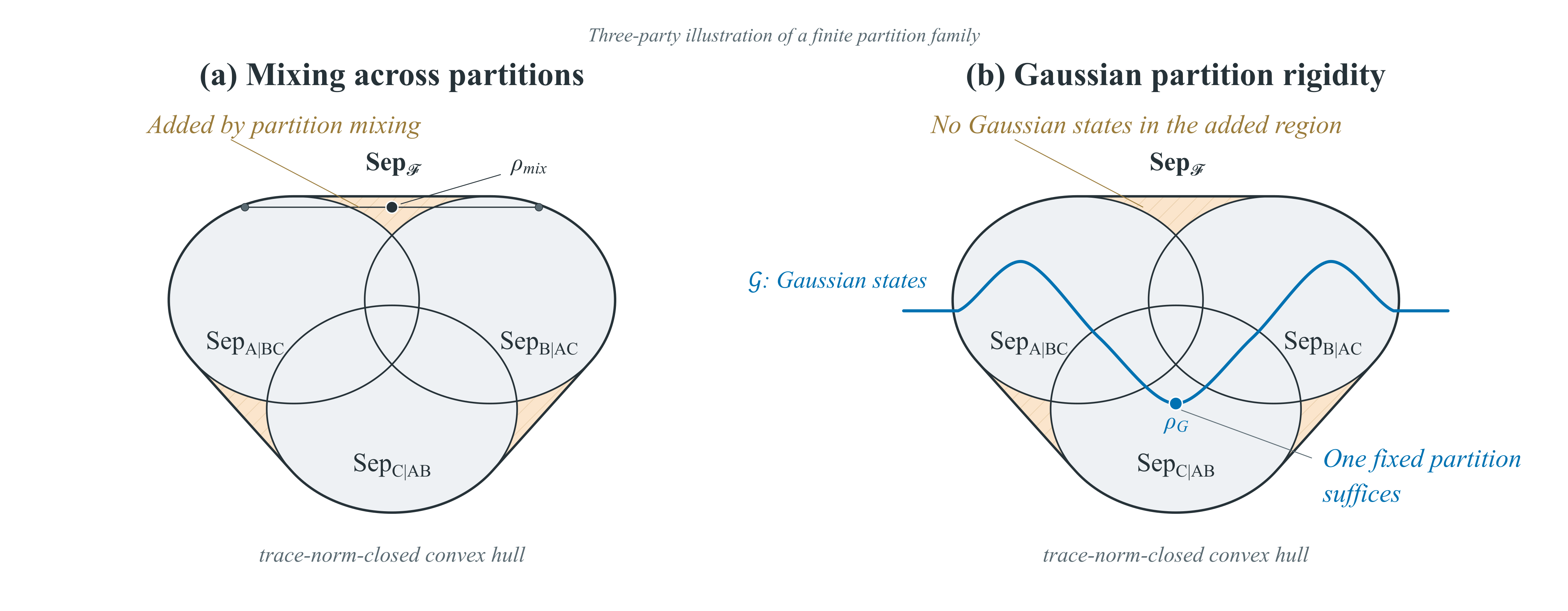}
  \caption{Convex partition rigidity for a three-party illustration with
  $\Ff=\{A|BC,B|AC,C|AB\}$. Mixing states that are separable across different cuts enlarges the unrestricted mixed-state class. Theorem~\ref{thm:main} says that the added region has empty intersection with the finite-mode Gaussian state set: every Gaussian member of $\Sep_{\Ff}$ already belongs to one fixed-partition separable set. The picture is schematic and the theorem holds for any finite family of partitions of finitely many finite-mode parties.}
  \label{fig:convex-partition-rigidity}
\end{figure*}

\begin{proof}
The inclusion from right to left is immediate from Eq.~\eqref{eq:family-sep}. For the converse, take $\rho_G\in\Gg\cap\Sep_{\Ff}$ and any sequence $0<s_k<1$ with $s_k\uparrow1$. Lemma~\ref{lem:Gaussian-powers} gives $\Tr\rho_G^{1-s_k}<\infty$. Lemma~\ref{lem:one-power-selector}, applied to $\rho_G$ with exponent $s_k$, therefore yields $\pi_k\in\Ff$ and a nonzero $\pi_k$-product vector
\begin{equation}
 \phi_k\in\Ran\rho_G^{s_k/2}.
\end{equation}
The normalized power $\rho_{G,s_k}=\rho_G^{s_k}/\Tr\rho_G^{s_k}$ is Gaussian and satisfies
\begin{equation}
 \Ran\rho_{G,s_k}^{1/2}=\Ran\rho_G^{s_k/2}
\end{equation}
by Lemma~\ref{lem:Gaussian-powers}. Theorem~\ref{thm:product-range} thus gives $\rho_{G,s_k}\in\Sep_{\pi_k}$. Since $\Ff$ is finite, some partition $\pi\in\Ff$ occurs along an infinite subsequence, say $\pi_{k_j}=\pi$. Along the same subsequence, Lemma~\ref{lem:Gaussian-powers} gives $\rho_{G,s_{k_j}}\to\rho_G$ in trace norm. Because $\Sep_\pi$ is trace-norm closed, $\rho_G\in\Sep_\pi$, proving Eq.~\eqref{eq:main-theorem}.
\end{proof}

\section{From convex decompositions to spectral ranges}
\label{sec:bridges}

\begin{figure*}[t]
  \centering
  \includegraphics[width=\textwidth]{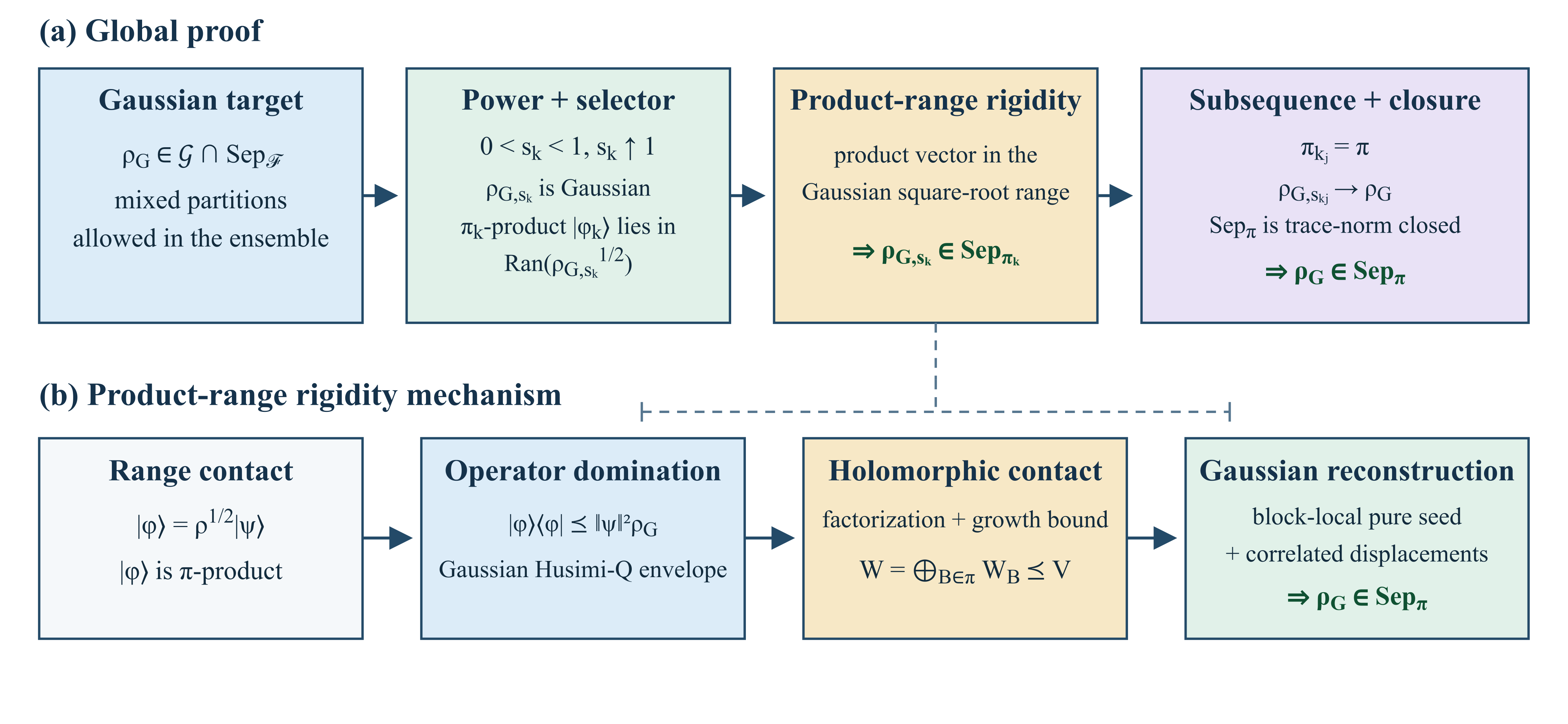}
  \caption{Proof roadmap for Theorem~\ref{thm:main}. (a) Gaussian powers and the spectral selector produce a product vector in a square-root range. Product-range rigidity fixes a separable partition for each power; finiteness of $\Ff$ and trace-norm closure then fix one partition for the target. (b) The local rigidity mechanism: operator domination gives a Gaussian Husimi-$Q$ envelope, holomorphic product contact yields a block-local pure covariance $W\preceq V$, and correlated displacements reconstruct a state separable across the same partition.}
  \label{fig:proof-roadmap}
\end{figure*}

\subsection{Spectral extraction of a product vector}

The main theorem starts from a convex decomposition that may mix several partitions, so no fixed cut is available a priori. We therefore look for a weaker object that can survive this mixing: one product vector in a spectral range of the target. The exponent $s<1$ is precisely what makes the extraction integrable; Gaussian powers will then turn this subcritical range into the square-root range required by product-range rigidity.

For a finite family $\Ff$, write
\begin{equation}
 \Pp_\Ff:=\bigcup_{\pi\in\Ff}\Pp_\pi.
 \label{eq:PF}
\end{equation}
Each $\Pp_\pi$ is trace-norm closed: if pure $\pi$-product projectors converge in trace norm, purity is continuous and partial trace is trace-norm contractive, so the limit is pure and every block marginal is pure, hence the limit is again $\pi$-product. The finite union $\Pp_\Ff$ is therefore closed. Since $\Sep_\Ff=\clconv(\Pp_\Ff)$, Lemma~1 of Holevo, Shirokov, and Werner \cite{HolevoShirokovWerner2005} gives a Borel probability measure $\nu$ supported on $\Pp_\Ff$ such that
\begin{equation}
 \rho=\int_{\Pp_\Ff}\Pi\,\nu(d\Pi)
 \label{eq:global-ensemble}
\end{equation}
for every $\rho\in\Sep_\Ff$, with the integral understood in the trace-class (Bochner) sense used for continuous quantum-state ensembles in Ref.~\cite{HolevoShirokovWerner2005}.

\begin{lemma}[One-power spectral selector]
\label{lem:one-power-selector}
Let $\rho\in\Sep_\Ff$ and $s\in(0,1)$ satisfy
\begin{equation}
 \Tr\rho^{1-s}<\infty.
 \label{eq:one-power-trace-assumption}
\end{equation}
Then there exist $\pi\in\Ff$ and a normalized $\pi$-product vector $|\phi\rangle$ such that
\begin{equation}
 |\phi\rangle\in\Ran\rho^{s/2}.
 \label{eq:one-powered-range}
\end{equation}
\end{lemma}

\begin{proof}
Use Eq.~\eqref{eq:global-ensemble} and write the spectral decomposition on $\supp\rho$ as
\begin{equation}
 \rho=\sum_j\lambda_j|e_j\rangle\langle e_j|,\qquad \lambda_j>0.
 \label{eq:spectral-rho-bridge}
\end{equation}
Let $P_0$ project onto $\ker\rho$. Since $P_0$ is bounded,
\begin{equation}
 \int\Tr(P_0\Pi)\,\nu(d\Pi)=\Tr(P_0\rho)=0,
\end{equation}
so almost every pure component lies in $\supp\rho$. All terms in the next expression are nonnegative; Tonelli's theorem therefore gives
\begin{align}
 &\int_{\Pp_\Ff}\sum_j\lambda_j^{-s}\langle e_j|\Pi|e_j\rangle\,\nu(d\Pi)\notag\\
 &\qquad=\sum_j\lambda_j^{-s}\langle e_j|\rho|e_j\rangle
 =\sum_j\lambda_j^{1-s}
 =\Tr\rho^{1-s}<\infty.
 \label{eq:spectral-bridge-identity}
\end{align}
Hence for $\nu$-almost every $\Pi=|\phi\rangle\langle\phi|$, writing $|\phi\rangle=\sum_jc_j|e_j\rangle$, one has
\begin{equation}
 \sum_j|c_j|^2\lambda_j^{-s}<\infty.
\end{equation}
Equivalently, $c_j=\lambda_j^{s/2}d_j$ for some $(d_j)\in\ell^2$, so
\begin{equation}
 |\phi\rangle\in\Ran\rho^{s/2}.
 \label{eq:spectral-range-test}
\end{equation}
Because this full-measure set is nonempty, choose any projector from it. Its membership in $\Pp_\Ff$ means that it is product across at least one $\pi\in\Ff$.
\end{proof}

The new content of the selector is not spectral calculus itself; it is the conversion of an unrestricted continuous separable ensemble into the existence of a product vector in a single subcritical power range.

\subsection{Gaussian powers as the interpolation step}

Normalized powers are not an auxiliary regularization here. They are the exact interpolation that turns the selector's subcritical range into the square-root range of another Gaussian state while keeping that state inside the Gaussian family. Taking $s\uparrow1$ then returns to the original target.

For faithful Gaussian states, normalized operator powers admit standard Gaussian formulas \cite{SeshadreesanLamiWilde2018}. The following form, proved directly from the Gaussian normal-mode decomposition, also includes nonfaithful boundary points.

\begin{lemma}[Gaussian powers and continuity]
\label{lem:Gaussian-powers}
Let $\rho_G$ be a finite-mode Gaussian state, possibly nonfaithful. For every $s>0$,
\begin{equation}
 \rho_{G,s}:=\frac{\rho_G^s}{\Tr\rho_G^s}
 \label{eq:normalized-power}
\end{equation}
is Gaussian, $\Tr\rho_G^s<\infty$, and
\begin{equation}
 \rho_{G,s}\longrightarrow\rho_G
 \quad\text{in trace norm as }s\uparrow1.
 \label{eq:power-continuity}
\end{equation}
Moreover,
\begin{equation}
 \Ran\rho_{G,s}^{1/2}=\Ran\rho_G^{s/2}.
 \label{eq:power-range-equality}
\end{equation}
\end{lemma}

\begin{proof}
The standard Gaussian normal-mode (Williamson) decomposition \cite{Weedbrook2012,Walschaers2021,SeshadreesanLamiWilde2018,Serafini2017} gives
\begin{equation}
 \rho_G=U_G\left[\bigotimes_{j=1}^m\tau(q_j)\right]U_G^\dagger,
 \qquad q_j\in[0,1),
 \label{eq:Williamson-density}
\end{equation}
where
\begin{equation}
 \tau(q)=(1-q)\sum_{n\ge0}q^n|n\rangle\langle n|\quad(q>0),
 \qquad \tau(0)=|0\rangle\langle0|.
\end{equation}
Direct functional calculus gives
\begin{equation}
 \frac{\tau(q)^s}{\Tr\tau(q)^s}=\tau(q^s),
 \label{eq:thermal-power}
\end{equation}
including $q=0$. Therefore
\begin{equation}
 \rho_{G,s}=U_G\left[\bigotimes_j\tau(q_j^s)\right]U_G^\dagger
 \label{eq:Gaussian-power-form}
\end{equation}
is Gaussian, while
\begin{equation}
 \Tr\rho_G^s
 =\prod_{j:q_j>0}\frac{(1-q_j)^s}{1-q_j^s}<\infty.
 \label{eq:Gaussian-power-trace}
\end{equation}
Thus pure Williamson modes cause no divergence and require no faithfulness-restoring noise.

For trace-norm continuity, let $\{\lambda_\ell\}$ be the nonzero eigenvalues of $\rho_G$ and fix $s_0\in(0,1)$. Equation~\eqref{eq:Gaussian-power-trace} gives $\sum_\ell\lambda_\ell^{s_0}<\infty$. For $s\in[s_0,1]$, both $\lambda_\ell^s$ and $|\lambda_\ell^s-\lambda_\ell|$ are dominated by summable multiples of $\lambda_\ell^{s_0}$, so dominated convergence yields $\Tr\rho_G^s\to1$ and $\|\rho_G^s-\rho_G\|_1\to0$. Normalization then gives Eq.~\eqref{eq:power-continuity}. Finally, multiplication by the positive scalar $(\Tr\rho_G^s)^{-1/2}$ does not change an operator range, proving Eq.~\eqref{eq:power-range-equality}.
\end{proof}

\section{Multipartition product-range rigidity}

For $\pi=\{B_1,\ldots,B_r\}$, define the set of (not necessarily normalized) $\pi$-product vectors
\begin{equation}
 \Seg_\pi:=\left\{\bigotimes_{\alpha=1}^r\phi_\alpha:\phi_\alpha\in\Hh_{B_\alpha}\right\}.
\end{equation}
This set is closed under multiplication by complex scalars; no convexity is implied by the notation.

\begin{theorem}[Multipartition product-range rigidity]
\label{thm:product-range}
Let $\rho_G$ be any finite-mode Gaussian state, possibly nonfaithful, and let $\pi$ be any partition of its physical parties. Then
\begin{equation}
 \Ran\rho_G^{1/2}\cap\Seg_\pi\neq\{0\}
 \quad\Longrightarrow\quad
 \rho_G\in\Sep_\pi.
 \label{eq:product-range-rigidity}
\end{equation}
\end{theorem}

For a generic mixed state, the presence of a single product vector in $\Ran\rho^{1/2}$ has no comparable global implication. The content of Theorem~\ref{thm:product-range} is that Gaussian analyticity makes this local range contact rigid: one product vector forces separability of the entire Gaussian state across the same partition.

\subsection{Coherent-state reduction and Gaussian holomorphic majorant}

First moments can be removed without loss of generality. If $U=D(d)=\bigotimes_\alpha D_{B_\alpha}(d_\alpha)$ centers $\rho_G$, then functional calculus gives $(U^\dagger\rho_GU)^{1/2}=U^\dagger\rho_G^{1/2}U$; hence both the product-range hypothesis and $\pi$-separability are preserved. We may therefore assume that $\rho_G$ is centered and that complete mode pairs are ordered block by block according to $\pi$.

Suppose
\begin{equation}
 0\neq|\phi\rangle=\rho_G^{1/2}|\psi\rangle\in\Seg_\pi.
\end{equation}
For every $|x\rangle$,
\begin{equation}
 |\langle x|\phi\rangle|^2
 =|\langle\rho_G^{1/2}x|\psi\rangle|^2
 \le\|\psi\|^2\langle x|\rho_G|x\rangle,
\end{equation}
so
\begin{equation}
 |\phi\rangle\langle\phi|\preceq\|\psi\|^2\rho_G.
 \label{eq:rank-one-domination}
\end{equation}
This is also the elementary operator-domination direction of Douglas' range-factorization theorem \cite{Douglas1966}; no inverse of $\rho_G$ is used.

For normalized multimode coherent states \cite{Glauber1963,Weedbrook2012,Walschaers2021}, use $a_j=(q_j+ip_j)/\sqrt2$, write $z_j=x_j+iy_j$, and set
\begin{equation}
 r(z)=(x_1,y_1,\ldots,x_m,y_m)^T,
 \qquad J=-\Omega.
\end{equation}
Then $\langle R\rangle_{|z\rangle}=\sqrt2\,r(z)$ and $|\langle0|z\rangle|^2=e^{-|z|^2}$ \cite{Glauber1963}. The corresponding Husimi-$Q$ diagonal---the coherent-state overlap $\langle z|\rho_G|z\rangle$ up to the conventional phase-space normalization---is Gaussian \cite{Husimi1940,Weedbrook2012,Walschaers2021,Serafini2017}; in the convention $V_{\rm vac}=\id$ it is
\begin{align}
 \langle z|\rho_G|z\rangle
 &=c_V\exp\!\left[-2r(z)^T(V+\id)^{-1}r(z)\right],\notag\\
 c_V&=\frac{2^m}{\sqrt{\det(V+\id)}}.
 \label{eq:husimi-diagonal}
\end{align}
The Gaussian-stripped amplitude
\begin{equation}
 \Phi(z):=e^{|z|^2/2}\langle\phi|z\rangle
 \label{eq:stripped-amplitude}
\end{equation}
is entire in the Segal--Bargmann (holomorphic) representation of bosonic quantum states \cite{ChabaudMehraban2022}. Since $|\phi\rangle$ is product over $\pi$ and coherent states factorize,
\begin{equation}
 \Phi(z)=\prod_{\alpha=1}^r f_\alpha(z_{B_\alpha}).
 \label{eq:coherent-product-factorization}
\end{equation}
Equations~\eqref{eq:rank-one-domination} and \eqref{eq:husimi-diagonal} give
\begin{equation}
 |\Phi(z)|^2\le M\exp\!\left[r(z)^TCr(z)\right],
 \qquad M=\|\psi\|^2c_V,
 \label{eq:Gaussian-envelope}
\end{equation}
where
\begin{equation}
 C=(V-\id)(V+\id)^{-1}
 =\id-2(V+\id)^{-1}.
 \label{eq:C-bounds}
\end{equation}
Because $V>0$, $C=C^T$ and
\begin{equation}
 -\id\prec C\prec\id.
\end{equation}

At this point the operator problem has become a finite-dimensional growth problem for an entire function. Operator domination supplies the Gaussian quadratic envelope, while the product vector supplies block factorization of the Bargmann amplitude. The next lemma shows that these two facts force the quadratic envelope to contain a block-local pure Gaussian envelope. Its matrix $D$ parametrizes the pure covariance reconstructed below.

\subsection{Block-holomorphic contact lemma}

\begin{lemma}[Block-holomorphic contact lemma]
\label{lem:block-contact}
Let $\mathbb R^{2m}=\bigoplus_{\alpha=1}^r\mathbb R^{2m_\alpha}$ be a decomposition into $J$-invariant mode blocks, with $J=-\Omega$ block diagonal in the same decomposition. Let $C=C^T$ satisfy $-\id\prec C\prec\id$. Let $\Phi$ be a nonzero entire function on $\mathbb C^m$ that factorizes as
\begin{equation}
 \Phi(z)=\prod_{\alpha=1}^r f_\alpha(z_{B_\alpha})
 \end{equation}
and obeys, for some $M>0$,
\begin{equation}
 |\Phi(z)|^2\le M e^{r(z)^TCr(z)}.
\end{equation}
Then there exists a real symmetric block-diagonal matrix
\begin{equation}
 D=\bigoplus_{\alpha=1}^rD_\alpha
\end{equation}
such that
\begin{equation}
 DJ+JD=0,
 \qquad
 -J^TCJ\preceq D\preceq C.
 \label{eq:D-limit-constraints}
\end{equation}
In particular,
\begin{equation}
 -\id\prec D\prec\id.
\end{equation}
\end{lemma}

\begin{proof}
For $\delta>0$ define the penalized function
\begin{equation}
 G_\delta(r)
 :=|\Phi(z(r))|^2
 \exp\!\left[-r^TCr-\delta|r|^2\right].
\end{equation}
The assumed majorant gives
\begin{equation}
 0\le G_\delta(r)\le M e^{-\delta|r|^2}.
\end{equation}
Thus $G_\delta$ is continuous, tends to zero as $|r|\to\infty$, and is not identically zero because $\Phi\not\equiv0$. It therefore attains a strictly positive global maximum at some finite $r_\delta$. In particular,
\begin{equation}
 \Phi(z(r_\delta))\neq0.
\end{equation}
Because the product of the factors is nonzero at this point, every factor $f_\alpha(z_{B_\alpha}(r_\delta))$ is nonzero. Each factor therefore admits a local holomorphic logarithm near its block component of the contact point.

Define locally
\begin{equation}
 u(r)=\log|\Phi(z(r))|
\end{equation}
and let
\begin{equation}
 D_\delta:=H_u(r_\delta)
\end{equation}
be its real Hessian. Since locally
\begin{equation}
 \log\Phi(z)=\sum_{\alpha=1}^r h_\alpha(z_{B_\alpha})
\end{equation}
for holomorphic $h_\alpha$, one has $u=\sum_\alpha\operatorname{Re}h_\alpha$. Hence all mixed second derivatives between distinct blocks vanish and $D_\delta$ is block diagonal.

For completeness, the Cauchy--Riemann step can be read directly from the Hessian. Within one block write $h=u+iv$ and temporarily order its real coordinates as $(x_1,\ldots,x_\ell,y_1,\ldots,y_\ell)$. Differentiating $u_{x_j}=v_{y_j}$ and $u_{y_j}=-v_{x_j}$ gives
\begin{equation}
 H_u=\begin{pmatrix}A&B\\ B&-A\end{pmatrix},
 \qquad
 J=\begin{pmatrix}0&-\id\\ \id&0\end{pmatrix},
\end{equation}
with $A=A^T$ and $B=B^T$. Hence $H_uJ+JH_u=0$. Returning to the interleaved real-coordinate order merely conjugates both matrices by the same permutation. Therefore, blockwise and hence globally,
\begin{equation}
 D_\delta J+JD_\delta=0.
 \label{eq:Ddelta-CR}
\end{equation}
Equivalently, using $J^T=-J$ and $J^2=-\id$,
\begin{equation}
 J^TD_\delta J=-D_\delta.
 \label{eq:JDJ}
\end{equation}

At the maximum of $G_\delta$, the real Hessian of $\log G_\delta$ is negative semidefinite. Since
\begin{equation}
 \log G_\delta(r)=2u(r)-r^TCr-\delta|r|^2,
\end{equation}
its Hessian at $r_\delta$ is
\begin{equation}
 2D_\delta-2C-2\delta\id\preceq0.
\end{equation}
Thus
\begin{equation}
 D_\delta\preceq C+\delta\id.
 \label{eq:Ddelta-upper}
\end{equation}
Conjugating Eq.~\eqref{eq:Ddelta-upper} by the orthogonal matrix $J$ and using Eq.~\eqref{eq:JDJ} gives
\begin{equation}
 -D_\delta\preceq J^TCJ+\delta\id,
\end{equation}
or
\begin{equation}
 -J^TCJ-\delta\id
 \preceq D_\delta\preceq C+\delta\id.
 \label{eq:Ddelta-two-sided}
\end{equation}

For $0<\delta\le1$, Eq.~\eqref{eq:Ddelta-two-sided} bounds all eigenvalues of $D_\delta$ uniformly. The space of real symmetric $2m\times2m$ matrices is finite dimensional, so for any sequence $\delta_\ell\downarrow0$ there is a convergent subsequence. Relabel it and write
\begin{equation}
 D_{\delta_\ell}\longrightarrow D,
 \qquad D=\bigoplus_{\alpha=1}^rD_\alpha.
 \label{eq:D-limit}
\end{equation}
Passing to the limit in Eqs.~\eqref{eq:Ddelta-CR} and \eqref{eq:Ddelta-two-sided} gives Eq.~\eqref{eq:D-limit-constraints}. Finally, $D\preceq C\prec\id$ implies $D\prec\id$, while $C\prec\id$ implies $J^TCJ\prec\id$ and hence
\begin{equation}
 D\succeq-J^TCJ\succ-\id.
\end{equation}
Therefore $-\id\prec D\prec\id$.
\end{proof}

The three outputs of the lemma have distinct roles: block diagonality encodes the prescribed partition, $DJ+JD=0$ becomes purity after the matrix transform used below, and $D\preceq C$ becomes covariance domination below the target.

\subsection{Completion of Theorem~\ref{thm:product-range}}

\begin{proof}[Proof of Theorem~\ref{thm:product-range}]
After the partition-local centering reduction, the product-range vector yields the majorant \eqref{eq:Gaussian-envelope}. Lemma~\ref{lem:block-contact} therefore gives a real symmetric block-local $D$ with
\begin{equation}
 DJ+JD=0,
 \qquad -\id\prec D\prec\id,
 \qquad D\preceq C.
\end{equation}
Define the fractional-linear matrix transform
\begin{equation}
 W=(\id+D)(\id-D)^{-1}
   =2(\id-D)^{-1}-\id.
 \label{eq:W-Cayley}
\end{equation}
Because $-\id\prec D\prec\id$, $W$ is positive and symmetric. Since $D$ is block diagonal, so is $W$. The anticommutation $DJ=-JD$ gives
\begin{equation}
 (\id-D)^{-1}J=J(\id+D)^{-1},
\end{equation}
whence $WJ=JW^{-1}$ and therefore
\begin{equation}
 WJW=J,
 \qquad\text{equivalently}\qquad
 W\Omega W=\Omega.
 \label{eq:W-symplectic}
\end{equation}
Thus $W$ is a pure Gaussian covariance \cite{SimonMukundaDutta1994,Weedbrook2012,Serafini2017}. Indeed, $WJW=J$ is equivalent to $WJ=JW^{-1}$; functional calculus for the positive matrix $W$ gives $W^{1/2}J=JW^{-1/2}$, and hence $W^{1/2}JW^{1/2}=J$. Therefore $S=W^{1/2}$ is block-local and symplectic, and $W=SS^T$ is the covariance of a pure Gaussian state product across $\pi$.

It remains to compare $W$ with $V$. From $D\preceq C$,
\begin{equation}
 \id-D\succeq\id-C\succ0.
\end{equation}
Inversion reverses the positive-semidefinite order for positive-definite matrices, with no commutativity assumption between $C$ and $D$, so
\begin{equation}
 (\id-D)^{-1}\preceq(\id-C)^{-1}.
\end{equation}
Using $C=(V-\id)(V+\id)^{-1}$, equivalently
\begin{equation}
 V=2(\id-C)^{-1}-\id,
\end{equation}
we obtain
\begin{equation}
 W\preceq V.
 \label{eq:W-below-V}
\end{equation}
Set $\Delta=V-W\succeq0$. Let $\mu_\Delta$ be a centered classical Gaussian distribution of displacements $d$ with
\begin{equation}
 \mathbb E_{\mu_\Delta}[dd^T]=\Delta/2.
\end{equation}
In the covariance convention \eqref{eq:covariance-convention}, averaging the block-product pure Gaussian seed of covariance $W$ over $D(d)=\bigotimes_\alpha D_{B_\alpha}(d_\alpha)$ adds $2\mathbb E[dd^T]=\Delta$ to the covariance. This is the correlated-displacement construction used in the Gaussian separability criterion of Werner and Wolf \cite{WernerWolf2001}. The resulting state is therefore the centered Gaussian state of covariance $V$ and is $\pi$-separable. Restoring the original first moments by the inverse partition-local displacement preserves $\pi$-separability. Hence $\rho_G\in\Sep_\pi$. Thus the product vector in the range is not merely a witness: it forces the target covariance to contain a pure product Gaussian core, with the remaining covariance generated by classical correlated displacement noise.
\end{proof}

\section{Consequences beyond bipartitions}
\label{sec:hierarchy}

The bipartition consequence was stated immediately after Theorem~\ref{thm:main}. The same rigidity collapses partition mixing throughout the usual multipartite hierarchy.

Let
\begin{equation}
 \mathcal S_k
 :=\clconv\!\left(\bigcup_{\pi:\,|\pi|\ge k}\Sep_\pi\right)
 \label{eq:ksep-def}
\end{equation}
be the usual $k$-separable class \cite{Gabriel2010,GuhneToth2009} (equivalently, one may use exactly $k$ blocks because refinement only strengthens separability),
\begin{equation}
 \mathcal S_k
 =\clconv\!\left(\bigcup_{\pi:\,|\pi|=k}\Sep_\pi\right).
 \label{eq:ksep-convention-equivalence}
\end{equation}
Then Theorem~\ref{thm:main} gives
\begin{corollary}[Fixed-partition $k$-separability]
\label{cor:ksep}
For every finite-mode Gaussian state,
\begin{equation}
 \boxed{
 \rho_G\in\mathcal S_k
 \Longleftrightarrow
 \exists\,\pi:\ |\pi|\ge k,\ \rho_G\in\Sep_\pi.}
 \label{eq:ksep-rigidity}
\end{equation}
\end{corollary}

For $k$-producibility \cite{SorensenMolmer2001,GuhneTothBriegel2005,Lu2018}, define the allowed partition family
\begin{equation}
 \Ff_k^{\rm prod}
 :=\left\{\pi\in\Part([n]):\max_{B\in\pi}|B|\le k\right\}
 \label{eq:prod-family}
\end{equation}
and
\begin{equation}
 \mathcal P_k
 :=\clconv\!\left(\bigcup_{\pi\in\Ff_k^{\rm prod}}\Sep_\pi\right).
 \label{eq:kprod-def}
\end{equation}
Here $k$ counts physical parties in a block, not bosonic modes inside one party. Again by Theorem~\ref{thm:main},
\begin{corollary}[Fixed-partition $k$-producibility and exact party-based depth]
\label{cor:kprod}
\label{cor:depth}
For every finite-mode Gaussian state,
\begin{equation}
 \boxed{
 \rho_G\in\mathcal P_k
 \Longleftrightarrow
 \exists\,\pi:\ \rho_G\in\Sep_\pi,\ \max_{B\in\pi}|B|\le k.}
 \label{eq:kprod-rigidity}
\end{equation}
Consequently, for the party-based depth
\begin{equation}
 D_{\rm ent}(\rho):=\min\{k:\rho\in\mathcal P_k\},
 \label{eq:depth-def}
\end{equation}
one has
\begin{equation}
 \boxed{
 D_{\rm ent}(\rho_G)
 =\min_{\pi:\,\rho_G\in\Sep_\pi}\max_{B\in\pi}|B|.}
 \label{eq:depth-formula}
\end{equation}
\end{corollary}

\begin{corollary}[Persistence under party-wise tensor powers]
\label{cor:tensor-power}
If $\rho_G\in\Sep_\Ff$, then there exists a fixed $\pi\in\Ff$ such that
\begin{equation}
 \rho_G^{\otimes N}\in\Sep_\pi
 \qquad\text{for every }N\ge1.
 \label{eq:no-activation}
\end{equation}
\end{corollary}

In particular, party-wise tensor powers cannot activate GME from a biseparable finite-mode Gaussian state. Here each physical party $A_j$ is understood to hold all copies of its local subsystem. This identifies a precise boundary within the broader multicopy resource theory. Activation is possible in finite-dimensional settings \cite{PalazuelosDeVicente2022,Yamasaki2022}, has been demonstrated experimentally using two copies of a biseparable three-qubit state \cite{Starek2026}, and extends to non-Gaussian continuous-variable states \cite{Baksova2025}. By contrast, for an exact biseparable finite-mode Gaussian state, Theorem~\ref{thm:main} supplies a fixed separable cut and Eq.~\eqref{eq:no-activation} shows that the same cut survives every party-wise tensor power. The conclusion does not cover arbitrary regrouping across copies.

\section{Fixed-partition covariance characterization and computation}
\label{sec:covariance}

The structural theorem also yields a direct computational consequence. Once partition mixing has been removed, deciding whether a Gaussian state belongs to $\Sep_\Ff$ reduces to checking finitely many fixed-partition covariance feasibility problems. For one prescribed partition the required criterion is the following, extending the bipartite and three-mode criteria of Refs.~\cite{WernerWolf2001,Giedke2001,GiedkeTripartite2001}.

\begin{proposition}[Fixed-partition Gaussian covariance criterion]
\label{thm:cov-criterion}
Let $\pi=\{B_1,\ldots,B_r\}$ and let $\rho_G(V)$ be a finite-mode Gaussian state. Then
\begin{equation}
 \rho_G(V)\in\Sep_\pi
 \label{eq:cov-sep-left}
\end{equation}
if and only if there are real symmetric block covariances $W_{B_\alpha}$ satisfying
\begin{equation}
 W_{B_\alpha}+i\Omega_{B_\alpha}\succeq0
 \qquad(\alpha=1,\ldots,r)
 \label{eq:local-physical-W}
\end{equation}
and
\begin{equation}
 V-\bigoplus_{\alpha=1}^rW_{B_\alpha}\succeq0.
 \label{eq:cov-domination}
\end{equation}
\end{proposition}

For context, Werner and Wolf proved this covariance-domination criterion as Proposition~1 for arbitrary-mode \emph{bipartite} Gaussian systems \cite{WernerWolf2001}, and Giedke \emph{et al.} obtained an independent arbitrary-mode bipartite criterion \cite{Giedke2001}. In the special $1\times1\times1$ three-mode setting, Theorem~$2'$ of Ref.~\cite{GiedkeTripartite2001} gives the three-party version. We are not aware of a primary source that states the same state-space covariance-domination theorem in this generality, for an arbitrary finite multipartition with arbitrary mode numbers in its blocks. The sufficiency direction extends directly, while Appendix~\ref{app:covnecessity} gives a self-contained necessity proof that accommodates continuous decompositions.

For sufficiency, put $W=\bigoplus_\alpha W_{B_\alpha}$ and $\Delta=V-W\succeq0$, choose centered Gaussian states $\sigma_\alpha$ with covariances $W_{B_\alpha}$, and choose a centered classical Gaussian displacement with covariance $\Delta/2$. Then
\begin{equation}
 \int\bigotimes_{\alpha=1}^r
 D_{B_\alpha}(d_\alpha)\sigma_\alpha
 D_{B_\alpha}(d_\alpha)^\dagger\,\mu_\Delta(dd)
 \label{eq:cov-suff-mixture}
\end{equation}
is $\pi$-separable and Gaussian with covariance $V$. This is the standard correlated-displacement reconstruction, with the factor $1/2$ fixed by the convention \eqref{eq:covariance-convention} \cite{Weedbrook2012,Serafini2017}. Appendix~\ref{app:covnecessity} proves necessity for an arbitrary finite multipartition directly from a continuous product ensemble and finite second moments.

For fixed $\pi$, Eqs.~\eqref{eq:local-physical-W}--\eqref{eq:cov-domination} are linear matrix inequalities. The Hermitian uncertainty condition has the real form
\begin{equation}
 W_B+i\Omega_B\succeq0
 \Longleftrightarrow
 \begin{pmatrix}W_B&-\Omega_B\\ \Omega_B&W_B\end{pmatrix}\succeq0.
 \label{eq:realified-uncertainty}
\end{equation}
Hence fixed-partition separability is a real semidefinite feasibility problem, and Theorem~\ref{thm:main} gives the finite disjunction
\begin{equation}
 \begin{aligned}
 \rho_G(V)\in\Sep_\Ff
 \quad\Longleftrightarrow\quad
 &\text{Eqs.~\eqref{eq:local-physical-W}--\eqref{eq:cov-domination}}\\
 &\text{are feasible for some }\pi\in\Ff.
 \end{aligned}
 \label{eq:family-SDP-disjunction}
\end{equation}
This exact state-space disjunction is not the same object as a single mixed-partition covariance relaxation \cite{HyllusEisert2006,Gerke2016,Baksova2025}.

\section{A mixed Gaussian example}
\label{sec:example}

To show that the hierarchy consequences are not restricted to pure or block-uncorrelated states, we give a strictly mixed four-party Gaussian state with cross-block classical correlations whose exact party-based entanglement depth is nevertheless fixed by one partition. Consider four one-mode parties in the interleaved ordering
\begin{equation}
 R=(q_1,p_1,q_2,p_2,q_3,p_3,q_4,p_4)^T.
\end{equation}
Let
\begin{equation}
 c=\cosh(2r),\qquad s=\sinh(2r),\qquad r>0,
\end{equation}
and let
\begin{equation}
 T_r=
 \begin{pmatrix}
 c&0&s&0\\
 0&c&0&-s\\
 s&0&c&0\\
 0&-s&0&c
 \end{pmatrix}
 \label{eq:TMSV-covariance}
\end{equation}
be a two-mode squeezed-vacuum covariance. The pure seed
\begin{equation}
 W_r=T_r\oplus T_r
 \label{eq:two-TMSV-seed}
\end{equation}
is product across $\pi_\star=\{\{1,2\},\{3,4\}\}$. Set
\begin{equation}
 \begin{aligned}
 u&=(1,0,0,0,1,0,0,0)^T,\\
 V_{r,t}&=W_r+tuu^T,\qquad t>0.
 \end{aligned}
 \label{eq:mixed-example-covariance}
\end{equation}
The added term is correlated classical displacement noise. If $\xi$ is centered Gaussian with variance $t/2$, displacing modes $1$ and $3$ by the same $q$ shift gives
\begin{equation}
 \rho_{r,t}=\int
 \left[D_1(\xi)\rho_{12}^{(r)}D_1(\xi)^\dagger\right]
 \otimes
 \left[D_3(\xi)\rho_{34}^{(r)}D_3(\xi)^\dagger\right]
 \,\mu_t(d\xi),
 \label{eq:mixed-example-decomposition}
\end{equation}
so
\begin{equation}
 \rho_{r,t}\in\Sep_{12|34}\subseteq\mathcal P_2.
 \label{eq:mixed-example-2prod}
\end{equation}
The target nevertheless has a cross-block classical correlation,
\begin{equation}
 (V_{r,t})_{q_1q_3}=t.
 \label{eq:mixed-cross-correlation}
\end{equation}

The state is strictly mixed (nonpure). Since $W_r$ is pure, the matrix determinant lemma gives
\begin{equation}
 \det V_{r,t}=1+t\,u^TW_r^{-1}u=1+2ct>1,
 \label{eq:mixed-example-determinant}
\end{equation}
so the Gaussian state is not pure \cite{Weedbrook2012,Serafini2017}.

To show that its exact party-based depth is two, reduce to modes $1$ and $2$:
\begin{equation}
 T_{r,t}^{(12)}=T_r+t e_{q_1}e_{q_1}^T.
 \label{eq:noisy-TMSV-reduction}
\end{equation}
In $2\times2$ block form,
\begin{equation}
 A=\operatorname{diag}(c+t,c),\qquad
 B=c\id_2,\qquad
 C_{12}=\operatorname{diag}(s,-s).
\end{equation}
For the partially transposed two-mode covariance, the squared symplectic eigenvalues are the roots of \cite{Simon2000,Giedke2001,Serafini2017}
\begin{equation}
 x^2-\widetilde\Delta x+\det T_{r,t}^{(12)}=0,
 \label{eq:PT-polynomial}
\end{equation}
with
\begin{equation}
 \widetilde\Delta=2(c^2+s^2)+ct,
 \qquad
 \det T_{r,t}^{(12)}=1+ct.
\end{equation}
At $x=1$ the polynomial equals
\begin{equation}
 1-\widetilde\Delta+\det T_{r,t}^{(12)}=-4s^2<0.
 \label{eq:PT-polynomial-at-one}
\end{equation}
Since the roots are positive, one satisfies $\widetilde\nu_-^2<1$; by the positive-partial-transpose criterion \cite{Peres1996,Simon2000}, the $12$ reduction is entangled for every $r,t>0$. The same holds for the $34$ reduction. Hence any four-party bipartition that separates either pair is impossible; every bipartition except $12|34$ separates at least one of the two pairs, while $12|34$ is explicitly feasible because $W_r\preceq V_{r,t}$. In particular the state is not fully separable, and together with Eq.~\eqref{eq:mixed-example-2prod} this gives
\begin{equation}
 \boxed{D_{\rm ent}(\rho_{r,t})=2.}
 \label{eq:mixed-example-depth}
\end{equation}

\section{Discussion}
\label{sec:discussion}

The central message is that the finite-mode Gaussian state set is rigid under convexification over separability partitions. For generic mixed states, the trace-norm-closed convex hull of several fixed-partition separable sets can contain states that belong to none of the sets individually. Theorem~\ref{thm:main} shows that this added convex region disappears when one restricts the target to be Gaussian. The bipartition instance closes the FIS--GME gap; the general statement shows that the same phenomenon is not specific to bipartitions but is a structural feature of the full partition hierarchy.

The mechanism also goes beyond a Gaussian-decomposition argument. The decomposition of the target may be continuous and arbitrarily non-Gaussian. What matters is instead the analytic structure of the target: operator domination becomes an exact quadratic growth envelope for a Bargmann entire function, while product structure becomes block factorization. Their contact forces a block-local pure Gaussian covariance below the target covariance. In this sense, Gaussianity removes convexification not by constraining the decomposition, but by making the target too rigid to occupy the region created solely by mixing partitions.

This rigidity has direct certification and computational consequences. For an exact Gaussian target, establishing inseparability across every bipartition certifies GME without an additional witness designed to exclude convex mixtures over different cuts. More generally, the mixed-partition membership problem in the full infinite-dimensional state space reduces to a finite disjunction of fixed-partition covariance semidefinite feasibility problems. The persistence of a fixed cut under party-wise tensor powers also separates the exact Gaussian sector from settings in which multiple copies of a biseparable state reveal GME \cite{PalazuelosDeVicente2022,Yamasaki2022,Baksova2025,Starek2026}.

The scope remains deliberately finite and exact: we assume finitely many parties, finitely many modes per party, an exact Gaussian target, and hence a finite partition family. Infinite-mode or thermodynamic limits may require new compactness and trace-class arguments, and approximately Gaussian targets would require a quantitative rigidity modulus. It is also essential to keep state-space biseparability distinct from covariance-only mixed-partition compatibility. Whether a comparable disappearance of convexification occurs on other structured analytic families of quantum states is an open question.

\begin{acknowledgments}
OpenAI ChatGPT (GPT-5.6 Sol) was used as an assistive tool for literature synthesis, mathematical proof auditing, and manuscript drafting and revision. The authors formulated the research questions and arguments, directed the AI-assisted analyses, independently checked the resulting mathematical and bibliographic content, and take full responsibility for all results and statements in the manuscript.
\end{acknowledgments}

\section*{Data Availability}
No data were created or analyzed in this study.

\appendix

\section{Necessity of the fixed-partition covariance criterion}
\label{app:covnecessity}

We give a self-contained proof of covariance necessity for an arbitrary finite partition under the trace-norm-closed state-space definition of separability. This extension is useful because the original Werner--Wolf formulation is stated for two blocks \cite{WernerWolf2001}, whereas here the fixed partition may contain several blocks and the separable decomposition may be continuous. We use throughout the phase-space and covariance convention of Eq.~\eqref{eq:covariance-convention}.

Let $\pi=\{B_1,\ldots,B_r\}$ be any partition into mode blocks. The fixed-partition separable class is the trace-norm-closed convex hull of pure $\pi$-product projectors. Since that set of projectors is closed, the barycentric representation theorem of Holevo, Shirokov, and Werner \cite{HolevoShirokovWerner2005} yields, for every $\rho\in\Sep_\pi$, a Borel probability measure $\nu$ supported on $\Pp_\pi$ such that
\begin{equation}
 \rho=\int_{\Pp_\pi}\Pi\,\nu(d\Pi)
 \label{eq:single-pi-ensemble}
\end{equation}
in the trace-class (Bochner) sense of Ref.~\cite{HolevoShirokovWerner2005}. Thus Eq.~\eqref{eq:single-pi-ensemble} reproduces the expectation of every bounded observable.

\subsection{First-order uncertainty without quadratic operator domains}

The covariance reconstruction below only needs vectors in the domains of the individual quadratures. We record the corresponding weak canonical-commutator fact so that no expectation value of an unbounded product $R_jR_k$ is silently assumed.

\begin{lemma}[Weak CCR and covariance uncertainty]
\label{lem:weak-CCR}
Let $|\psi\rangle$ be normalized and belong to $\bigcap_{a=1}^{2m}\Dom R_a$. Then
\begin{equation}
 \langle R_a\psi,R_b\psi\rangle-
 \langle R_b\psi,R_a\psi\rangle=i\Omega_{ab}.
 \label{eq:weak-CCR}
\end{equation}
If $m_j=\langle\psi|R_j|\psi\rangle$ and
\begin{equation}
 Q_\psi(x)=2\left\|\sum_jx_j(R_j-m_j)|\psi\rangle\right\|^2,
 \qquad x\in\mathbb R^{2m},
 \label{eq:weak-quadratic-form}
\end{equation}
has real symmetric representing matrix $V_\psi$, then
\begin{equation}
 V_\psi+i\Omega\succeq0.
 \label{eq:weak-form-uncertainty}
\end{equation}
\end{lemma}

\begin{proof}
By the finite-mode Stone--von Neumann theorem it is enough to use the Schr\"odinger realization \cite{vonNeumann1931,Folland1989}. The $Q$--$Q$ identity is immediate and the $P$--$P$ identity follows after Fourier transform. For $Q_j=x_j$ and $P_k=-i\partial_k$, take a smooth cutoff $\chi_R$ equal to one on $|x|\le R$ and zero on $|x|\ge2R$. The assumptions $x_j\psi,\partial_k\psi\in L^2$ imply $x_j\bar\psi\,\partial_k\psi\in L^1$. Distributional integration by parts gives
\begin{equation}
 \int \chi_R x_j\partial_k|\psi|^2
 =-\delta_{jk}\int\chi_R|\psi|^2
  -\int x_j(\partial_k\chi_R)|\psi|^2.
 \label{eq:cutoff-ibp}
\end{equation}
The last term tends to zero because $|x_j\partial_k\chi_R|$ is uniformly bounded on the annulus $R\lesssim |x|\lesssim2R$, whereas the $L^2$ mass of $\psi$ on that annulus tends to zero. Dominated convergence on the left then yields $\int x_j\partial_k|\psi|^2=-\delta_{jk}$, which is exactly Eq.~\eqref{eq:weak-CCR} for $Q_j,P_k$.

Now set $a_j=(R_j-m_j)|\psi\rangle$ and $\Gamma_{jk}=\langle a_j|a_k\rangle$. The Gram matrix $\Gamma$ is positive semidefinite, Eq.~\eqref{eq:weak-CCR} gives $\Gamma-\Gamma^T=i\Omega$, and Eq.~\eqref{eq:weak-quadratic-form} gives $V_\psi=2\operatorname{Re}\Gamma$. Hence
\begin{equation}
 V_\psi+i\Omega=2\Gamma\succeq0.
\end{equation}
\end{proof}

\subsection{Continuous-ensemble law of total covariance}

\begin{proposition}[Multipartition covariance necessity]
\label{prop:multipartition-cov-necessity}
Let $\rho$ be any state with finite quadrature second moments that is separable across the finite partition $\pi=\{B_1,\ldots,B_r\}$. If $V$ is its covariance in the convention \eqref{eq:covariance-convention}, then there exist real symmetric block covariances $W_{B_\alpha}$ such that
\begin{equation}
 W_{B_\alpha}+i\Omega_{B_\alpha}\succeq0
 \quad(\alpha=1,\ldots,r),
 \label{eq:app-local-physical-W}
\end{equation}
and
\begin{equation}
 V-\bigoplus_{\alpha=1}^rW_{B_\alpha}\succeq0.
 \label{eq:app-cov-domination}
\end{equation}
\end{proposition}

\begin{proof}
Use the continuous product ensemble \eqref{eq:single-pi-ensemble}. Put $d=2m$ and consider the finite polarization set
\begin{equation}
 \mathcal X=\{e_j:1\le j\le d\}\cup
 \{e_j+e_k:1\le j<k\le d\}.
 \label{eq:polarization-set}
\end{equation}
For $x\in\mathcal X$, let $X_x=x^TR$ and $h_M(t)=\min\{t^2,M\}$. Since $h_M(X_x)$ is bounded and positive, Eq.~\eqref{eq:single-pi-ensemble} gives
\begin{equation}
 \Tr[\rho h_M(X_x)]
 =\int\Tr[\Pi h_M(X_x)]\,\nu(d\Pi).
\end{equation}
Monotone convergence and the finite second moment of $\rho$ imply
\begin{equation}
 \Tr(\rho X_x^2)
 =\int\langle X_x^2\rangle_\Pi\,\nu(d\Pi)<\infty.
 \label{eq:second-moment-disintegration}
\end{equation}
Because $\mathcal X$ is finite, there is one full-measure set $E\subseteq\Pp_\pi$ on which all these second moments are finite. In particular every representative vector of every $\Pi\in E$ lies in the domain of each $R_j$.

For $\Pi\in E$, define $m_j(\Pi)=\langle R_j\rangle_\Pi$ and the centered directional form
\begin{equation}
 Q_\Pi(x)=2\bigl(\langle X_x^2\rangle_\Pi-
 \langle X_x\rangle_\Pi^2\bigr)=x^TV(\Pi)x.
 \label{eq:directional-covariance-form}
\end{equation}
Its entries are recovered from the finite set \eqref{eq:polarization-set}, e.g.
\begin{align}
 V_{jj}(\Pi)&=Q_\Pi(e_j),\notag\\
 2V_{jk}(\Pi)&=Q_\Pi(e_j+e_k)-Q_\Pi(e_j)-Q_\Pi(e_k).
 \label{eq:covariance-polarization}
\end{align}
These functions of $\Pi$ are measurable because the unbounded moments above are monotone limits of bounded spectral truncations. Lemma~\ref{lem:weak-CCR} applies pointwise and yields
\begin{equation}
 V(\Pi)+i\Omega\succeq0.
 \label{eq:component-uncertainty}
\end{equation}

Product structure forces $V(\Pi)$ to be block diagonal across $\pi$. Indeed, for quadratures $X$ and $Y$ supported on distinct product blocks, the inner product $\langle X\psi_\Pi,Y\psi_\Pi\rangle$ factorizes as $\langle X\rangle_\Pi\langle Y\rangle_\Pi$, so the centered cross term in $Q_\Pi(X+Y)$ vanishes. Consequently
\begin{equation}
 V(\Pi)=\bigoplus_{\alpha=1}^r V_{B_\alpha}(\Pi),
 \qquad
 V_{B_\alpha}(\Pi)+i\Omega_{B_\alpha}\succeq0.
 \label{eq:component-block-covariance}
\end{equation}

It remains to average these pointwise covariances. Cauchy--Schwarz and Eq.~\eqref{eq:second-moment-disintegration} give $m_x(\Pi)=x^Tm(\Pi)\in L^2(\nu)$ for all $x\in\mathcal X$. Introduce the bounded clipping function $g_M(t)=\max\{-M,\min\{t,M\}\}$. Pointwise spectral convergence gives $\langle g_M(X_x)\rangle_\Pi\to m_x(\Pi)$, while
\begin{equation}
 |\langle g_M(X_x)\rangle_\Pi|
 \le \langle |X_x|\rangle_\Pi
 \le \langle X_x^2\rangle_\Pi^{1/2}.
\end{equation}
The last quantity is $L^1(\nu)$ because Eq.~\eqref{eq:second-moment-disintegration} and Cauchy--Schwarz give $\int\langle X_x^2\rangle_\Pi^{1/2}d\nu\le[\Tr(\rho X_x^2)]^{1/2}$. Dominated convergence applied to Eq.~\eqref{eq:single-pi-ensemble} therefore gives
\begin{equation}
 \langle X_x\rangle_\rho=\int m_x(\Pi)\,\nu(d\Pi).
 \label{eq:first-moment-disintegration}
\end{equation}
Combining this with Eq.~\eqref{eq:second-moment-disintegration} gives, for $x\in\mathcal X$,
\begin{equation}
 x^TVx=
 \int x^TV(\Pi)x\,\nu(d\Pi)
 +2\operatorname{Var}_\nu[m_x(\Pi)].
 \label{eq:directional-total-covariance}
\end{equation}
Both sides are quadratic forms, and the polarization set determines every matrix entry; hence
\begin{equation}
 V=\int V(\Pi)\,\nu(d\Pi)
 +2\operatorname{Cov}_\nu[m(\Pi)].
 \label{eq:total-covariance}
\end{equation}
The factor two is exactly the convention in Eq.~\eqref{eq:covariance-convention}.

Finally set
\begin{equation}
 W_{B_\alpha}=\int V_{B_\alpha}(\Pi)\,\nu(d\Pi).
 \label{eq:averaged-local-covariance}
\end{equation}
The diagonal entries are integrable by Eq.~\eqref{eq:second-moment-disintegration}; off-diagonal entries are controlled by polarization (or Cauchy--Schwarz for the centered form). Integrating Eq.~\eqref{eq:component-block-covariance} yields Eq.~\eqref{eq:app-local-physical-W}, while Eqs.~\eqref{eq:component-block-covariance} and \eqref{eq:total-covariance} yield
\begin{equation}
 V-\bigoplus_{\alpha=1}^rW_{B_\alpha}
 =2\operatorname{Cov}_\nu[m(\Pi)]\succeq0,
\end{equation}
which is Eq.~\eqref{eq:app-cov-domination}.
\end{proof}

For Gaussian $\rho_G(V)$, the converse is the standard correlated-displacement construction: if $W=\bigoplus_\alpha W_{B_\alpha}\preceq V$, take a product Gaussian seed with covariance $W$ and average its phase-space translates against a centered classical Gaussian distribution with covariance $(V-W)/2$. The resulting Gaussian state has covariance $V$ and is separable across $\pi$; this is the direct multipartition version of the converse construction in Proposition~1 of Werner and Wolf \cite{WernerWolf2001}.

\bibliography{manuscript}

\begin{thebibliography}{36}%
\makeatletter
\providecommand \@ifxundefined [1]{%
 \@ifx{#1\undefined}
}%
\providecommand \@ifnum [1]{%
 \ifnum #1\expandafter \@firstoftwo
 \else \expandafter \@secondoftwo
 \fi
}%
\providecommand \@ifx [1]{%
 \ifx #1\expandafter \@firstoftwo
 \else \expandafter \@secondoftwo
 \fi
}%
\providecommand \natexlab [1]{#1}%
\providecommand \enquote  [1]{``#1''}%
\providecommand \bibnamefont  [1]{#1}%
\providecommand \bibfnamefont [1]{#1}%
\providecommand \citenamefont [1]{#1}%
\providecommand \href@noop [0]{\@secondoftwo}%
\providecommand \href [0]{\begingroup \@sanitize@url \@href}%
\providecommand \@href[1]{\@@startlink{#1}\@@href}%
\providecommand \@@href[1]{\endgroup#1\@@endlink}%
\providecommand \@sanitize@url [0]{\catcode `\\12\catcode `\$12\catcode
  `\&12\catcode `\#12\catcode `\^12\catcode `\_12\catcode `\%12\relax}%
\providecommand \@@startlink[1]{}%
\providecommand \@@endlink[0]{}%
\providecommand \url  [0]{\begingroup\@sanitize@url \@url }%
\providecommand \@url [1]{\endgroup\@href {#1}{\urlprefix }}%
\providecommand \urlprefix  [0]{URL }%
\providecommand \Eprint [0]{\href }%
\providecommand \doibase [0]{https://doi.org/}%
\providecommand \selectlanguage [0]{\@gobble}%
\providecommand \bibinfo  [0]{\@secondoftwo}%
\providecommand \bibfield  [0]{\@secondoftwo}%
\providecommand \translation [1]{[#1]}%
\providecommand \BibitemOpen [0]{}%
\providecommand \bibitemStop [0]{}%
\providecommand \bibitemNoStop [0]{.\EOS\space}%
\providecommand \EOS [0]{\spacefactor3000\relax}%
\providecommand \BibitemShut  [1]{\csname bibitem#1\endcsname}%
\let\auto@bib@innerbib\@empty
\bibitem [{\citenamefont {D{\"u}r}\ \emph {et~al.}(1999)\citenamefont
  {D{\"u}r}, \citenamefont {Cirac},\ and\ \citenamefont
  {Tarrach}}]{DurCiracTarrach1999}%
  \BibitemOpen
  \bibfield  {author} {\bibinfo {author} {\bibfnamefont {W.}~\bibnamefont
  {D{\"u}r}}, \bibinfo {author} {\bibfnamefont {J.~I.}\ \bibnamefont {Cirac}},\
  and\ \bibinfo {author} {\bibfnamefont {R.}~\bibnamefont {Tarrach}},\
  }\bibfield  {title} {\bibinfo {title} {Separability and distillability of
  multiparticle quantum systems},\ }\href
  {https://doi.org/10.1103/PhysRevLett.83.3562} {\bibfield  {journal} {\bibinfo
   {journal} {Phys. Rev. Lett.}\ }\textbf {\bibinfo {volume} {83}},\ \bibinfo
  {pages} {3562} (\bibinfo {year} {1999})}\BibitemShut {NoStop}%
\bibitem [{\citenamefont {Horodecki}\ \emph {et~al.}(2009)\citenamefont
  {Horodecki}, \citenamefont {Horodecki}, \citenamefont {Horodecki},\ and\
  \citenamefont {Horodecki}}]{Horodecki2009}%
  \BibitemOpen
  \bibfield  {author} {\bibinfo {author} {\bibfnamefont {R.}~\bibnamefont
  {Horodecki}}, \bibinfo {author} {\bibfnamefont {P.}~\bibnamefont
  {Horodecki}}, \bibinfo {author} {\bibfnamefont {M.}~\bibnamefont
  {Horodecki}},\ and\ \bibinfo {author} {\bibfnamefont {K.}~\bibnamefont
  {Horodecki}},\ }\bibfield  {title} {\bibinfo {title} {Quantum entanglement},\
  }\href {https://doi.org/10.1103/RevModPhys.81.865} {\bibfield  {journal}
  {\bibinfo  {journal} {Rev. Mod. Phys.}\ }\textbf {\bibinfo {volume} {81}},\
  \bibinfo {pages} {865} (\bibinfo {year} {2009})}\BibitemShut {NoStop}%
\bibitem [{\citenamefont {G{\"u}hne}\ and\ \citenamefont
  {T{\'o}th}(2009)}]{GuhneToth2009}%
  \BibitemOpen
  \bibfield  {author} {\bibinfo {author} {\bibfnamefont {O.}~\bibnamefont
  {G{\"u}hne}}\ and\ \bibinfo {author} {\bibfnamefont {G.}~\bibnamefont
  {T{\'o}th}},\ }\bibfield  {title} {\bibinfo {title} {Entanglement
  detection},\ }\href {https://doi.org/10.1016/j.physrep.2009.02.004}
  {\bibfield  {journal} {\bibinfo  {journal} {Phys. Rep.}\ }\textbf {\bibinfo
  {volume} {474}},\ \bibinfo {pages} {1} (\bibinfo {year} {2009})}\BibitemShut
  {NoStop}%
\bibitem [{\citenamefont {Gabriel}\ \emph {et~al.}(2010)\citenamefont
  {Gabriel}, \citenamefont {Hiesmayr},\ and\ \citenamefont
  {Huber}}]{Gabriel2010}%
  \BibitemOpen
  \bibfield  {author} {\bibinfo {author} {\bibfnamefont {A.}~\bibnamefont
  {Gabriel}}, \bibinfo {author} {\bibfnamefont {B.~C.}\ \bibnamefont
  {Hiesmayr}},\ and\ \bibinfo {author} {\bibfnamefont {M.}~\bibnamefont
  {Huber}},\ }\bibfield  {title} {\bibinfo {title} {Criterion for
  {$k$}-separability in mixed multipartite systems},\ }\href
  {https://doi.org/10.26421/QIC10.9-10-8} {\bibfield  {journal} {\bibinfo
  {journal} {Quantum Inf. Comput.}\ }\textbf {\bibinfo {volume} {10}},\
  \bibinfo {pages} {829} (\bibinfo {year} {2010})}\BibitemShut {NoStop}%
\bibitem [{\citenamefont {G{\"u}hne}\ \emph {et~al.}(2005)\citenamefont
  {G{\"u}hne}, \citenamefont {T{\'o}th},\ and\ \citenamefont
  {Briegel}}]{GuhneTothBriegel2005}%
  \BibitemOpen
  \bibfield  {author} {\bibinfo {author} {\bibfnamefont {O.}~\bibnamefont
  {G{\"u}hne}}, \bibinfo {author} {\bibfnamefont {G.}~\bibnamefont
  {T{\'o}th}},\ and\ \bibinfo {author} {\bibfnamefont {H.~J.}\ \bibnamefont
  {Briegel}},\ }\bibfield  {title} {\bibinfo {title} {Multipartite entanglement
  in spin chains},\ }\href {https://doi.org/10.1088/1367-2630/7/1/229}
  {\bibfield  {journal} {\bibinfo  {journal} {New J. Phys.}\ }\textbf {\bibinfo
  {volume} {7}},\ \bibinfo {pages} {229} (\bibinfo {year} {2005})}\BibitemShut
  {NoStop}%
\bibitem [{\citenamefont {S{\o}rensen}\ and\ \citenamefont
  {M{\o}lmer}(2001)}]{SorensenMolmer2001}%
  \BibitemOpen
  \bibfield  {author} {\bibinfo {author} {\bibfnamefont {A.~S.}\ \bibnamefont
  {S{\o}rensen}}\ and\ \bibinfo {author} {\bibfnamefont {K.}~\bibnamefont
  {M{\o}lmer}},\ }\bibfield  {title} {\bibinfo {title} {Entanglement and
  extreme spin squeezing},\ }\href
  {https://doi.org/10.1103/PhysRevLett.86.4431} {\bibfield  {journal} {\bibinfo
   {journal} {Phys. Rev. Lett.}\ }\textbf {\bibinfo {volume} {86}},\ \bibinfo
  {pages} {4431} (\bibinfo {year} {2001})}\BibitemShut {NoStop}%
\bibitem [{\citenamefont {Lu}\ \emph {et~al.}(2018)\citenamefont {Lu},
  \citenamefont {Zhao}, \citenamefont {Li}, \citenamefont {Yin}, \citenamefont
  {Yuan}, \citenamefont {Hung}, \citenamefont {Chen}, \citenamefont {Li},
  \citenamefont {Liu}, \citenamefont {Peng}, \citenamefont {Liang},
  \citenamefont {Ma}, \citenamefont {Chen},\ and\ \citenamefont
  {Pan}}]{Lu2018}%
  \BibitemOpen
  \bibfield  {author} {\bibinfo {author} {\bibfnamefont {H.}~\bibnamefont
  {Lu}}, \bibinfo {author} {\bibfnamefont {Q.}~\bibnamefont {Zhao}}, \bibinfo
  {author} {\bibfnamefont {Z.-D.}\ \bibnamefont {Li}}, \bibinfo {author}
  {\bibfnamefont {X.-F.}\ \bibnamefont {Yin}}, \bibinfo {author} {\bibfnamefont
  {X.}~\bibnamefont {Yuan}}, \bibinfo {author} {\bibfnamefont {J.-C.}\
  \bibnamefont {Hung}}, \bibinfo {author} {\bibfnamefont {L.-K.}\ \bibnamefont
  {Chen}}, \bibinfo {author} {\bibfnamefont {L.}~\bibnamefont {Li}}, \bibinfo
  {author} {\bibfnamefont {N.-L.}\ \bibnamefont {Liu}}, \bibinfo {author}
  {\bibfnamefont {C.-Z.}\ \bibnamefont {Peng}}, \bibinfo {author}
  {\bibfnamefont {Y.-C.}\ \bibnamefont {Liang}}, \bibinfo {author}
  {\bibfnamefont {X.}~\bibnamefont {Ma}}, \bibinfo {author} {\bibfnamefont
  {Y.-A.}\ \bibnamefont {Chen}},\ and\ \bibinfo {author} {\bibfnamefont
  {J.-W.}\ \bibnamefont {Pan}},\ }\bibfield  {title} {\bibinfo {title}
  {Entanglement structure: Entanglement partitioning in multipartite systems
  and its experimental detection using optimizable witnesses},\ }\href
  {https://doi.org/10.1103/PhysRevX.8.021072} {\bibfield  {journal} {\bibinfo
  {journal} {Phys. Rev. X}\ }\textbf {\bibinfo {volume} {8}},\ \bibinfo {pages}
  {021072} (\bibinfo {year} {2018})}\BibitemShut {NoStop}%
\bibitem [{\citenamefont {Weedbrook}\ \emph {et~al.}(2012)\citenamefont
  {Weedbrook}, \citenamefont {Pirandola}, \citenamefont
  {Garc{\'i}a-Patr{\'o}n}, \citenamefont {Cerf}, \citenamefont {Ralph},
  \citenamefont {Shapiro},\ and\ \citenamefont {Lloyd}}]{Weedbrook2012}%
  \BibitemOpen
  \bibfield  {author} {\bibinfo {author} {\bibfnamefont {C.}~\bibnamefont
  {Weedbrook}}, \bibinfo {author} {\bibfnamefont {S.}~\bibnamefont
  {Pirandola}}, \bibinfo {author} {\bibfnamefont {R.}~\bibnamefont
  {Garc{\'i}a-Patr{\'o}n}}, \bibinfo {author} {\bibfnamefont {N.~J.}\
  \bibnamefont {Cerf}}, \bibinfo {author} {\bibfnamefont {T.~C.}\ \bibnamefont
  {Ralph}}, \bibinfo {author} {\bibfnamefont {J.~H.}\ \bibnamefont {Shapiro}},\
  and\ \bibinfo {author} {\bibfnamefont {S.}~\bibnamefont {Lloyd}},\ }\bibfield
   {title} {\bibinfo {title} {Gaussian quantum information},\ }\href
  {https://doi.org/10.1103/RevModPhys.84.621} {\bibfield  {journal} {\bibinfo
  {journal} {Rev. Mod. Phys.}\ }\textbf {\bibinfo {volume} {84}},\ \bibinfo
  {pages} {621} (\bibinfo {year} {2012})}\BibitemShut {NoStop}%
\bibitem [{\citenamefont {Walschaers}(2021)}]{Walschaers2021}%
  \BibitemOpen
  \bibfield  {author} {\bibinfo {author} {\bibfnamefont {M.}~\bibnamefont
  {Walschaers}},\ }\bibfield  {title} {\bibinfo {title} {Non-gaussian quantum
  states and where to find them},\ }\href
  {https://doi.org/10.1103/PRXQuantum.2.030204} {\bibfield  {journal} {\bibinfo
   {journal} {PRX Quantum}\ }\textbf {\bibinfo {volume} {2}},\ \bibinfo {pages}
  {030204} (\bibinfo {year} {2021})}\BibitemShut {NoStop}%
\bibitem [{\citenamefont {Serafini}(2017)}]{Serafini2017}%
  \BibitemOpen
  \bibfield  {author} {\bibinfo {author} {\bibfnamefont {A.}~\bibnamefont
  {Serafini}},\ }\href@noop {} {\emph {\bibinfo {title} {Quantum Continuous
  Variables: A Primer of Theoretical Methods}}}\ (\bibinfo  {publisher} {CRC
  Press},\ \bibinfo {address} {Boca Raton},\ \bibinfo {year}
  {2017})\BibitemShut {NoStop}%
\bibitem [{\citenamefont {Werner}\ and\ \citenamefont
  {Wolf}(2001)}]{WernerWolf2001}%
  \BibitemOpen
  \bibfield  {author} {\bibinfo {author} {\bibfnamefont {R.~F.}\ \bibnamefont
  {Werner}}\ and\ \bibinfo {author} {\bibfnamefont {M.~M.}\ \bibnamefont
  {Wolf}},\ }\bibfield  {title} {\bibinfo {title} {Bound entangled gaussian
  states},\ }\href {https://doi.org/10.1103/PhysRevLett.86.3658} {\bibfield
  {journal} {\bibinfo  {journal} {Phys. Rev. Lett.}\ }\textbf {\bibinfo
  {volume} {86}},\ \bibinfo {pages} {3658} (\bibinfo {year}
  {2001})}\BibitemShut {NoStop}%
\bibitem [{\citenamefont {Giedke}\ \emph
  {et~al.}(2001{\natexlab{a}})\citenamefont {Giedke}, \citenamefont {Kraus},
  \citenamefont {Lewenstein},\ and\ \citenamefont {Cirac}}]{Giedke2001}%
  \BibitemOpen
  \bibfield  {author} {\bibinfo {author} {\bibfnamefont {G.}~\bibnamefont
  {Giedke}}, \bibinfo {author} {\bibfnamefont {B.}~\bibnamefont {Kraus}},
  \bibinfo {author} {\bibfnamefont {M.}~\bibnamefont {Lewenstein}},\ and\
  \bibinfo {author} {\bibfnamefont {J.~I.}\ \bibnamefont {Cirac}},\ }\bibfield
  {title} {\bibinfo {title} {Entanglement criteria for all bipartite gaussian
  states},\ }\href {https://doi.org/10.1103/PhysRevLett.87.167904} {\bibfield
  {journal} {\bibinfo  {journal} {Phys. Rev. Lett.}\ }\textbf {\bibinfo
  {volume} {87}},\ \bibinfo {pages} {167904} (\bibinfo {year}
  {2001}{\natexlab{a}})}\BibitemShut {NoStop}%
\bibitem [{\citenamefont {Giedke}\ \emph
  {et~al.}(2001{\natexlab{b}})\citenamefont {Giedke}, \citenamefont {Kraus},
  \citenamefont {Lewenstein},\ and\ \citenamefont
  {Cirac}}]{GiedkeTripartite2001}%
  \BibitemOpen
  \bibfield  {author} {\bibinfo {author} {\bibfnamefont {G.}~\bibnamefont
  {Giedke}}, \bibinfo {author} {\bibfnamefont {B.}~\bibnamefont {Kraus}},
  \bibinfo {author} {\bibfnamefont {M.}~\bibnamefont {Lewenstein}},\ and\
  \bibinfo {author} {\bibfnamefont {J.~I.}\ \bibnamefont {Cirac}},\ }\bibfield
  {title} {\bibinfo {title} {Separability properties of three-mode gaussian
  states},\ }\href {https://doi.org/10.1103/PhysRevA.64.052303} {\bibfield
  {journal} {\bibinfo  {journal} {Phys. Rev. A}\ }\textbf {\bibinfo {volume}
  {64}},\ \bibinfo {pages} {052303} (\bibinfo {year}
  {2001}{\natexlab{b}})}\BibitemShut {NoStop}%
\bibitem [{\citenamefont {Holevo}\ \emph {et~al.}(2005)\citenamefont {Holevo},
  \citenamefont {Shirokov},\ and\ \citenamefont
  {Werner}}]{HolevoShirokovWerner2005}%
  \BibitemOpen
  \bibfield  {author} {\bibinfo {author} {\bibfnamefont {A.~S.}\ \bibnamefont
  {Holevo}}, \bibinfo {author} {\bibfnamefont {M.~E.}\ \bibnamefont
  {Shirokov}},\ and\ \bibinfo {author} {\bibfnamefont {R.~F.}\ \bibnamefont
  {Werner}},\ }\href@noop {} {\bibinfo {title} {Separability and
  entanglement-breaking in infinite dimensions}} (\bibinfo {year} {2005}),\
  \Eprint {https://arxiv.org/abs/quant-ph/0504204} {arXiv:quant-ph/0504204
  [quant-ph]} \BibitemShut {NoStop}%
\bibitem [{\citenamefont {Hyllus}\ and\ \citenamefont
  {Eisert}(2006)}]{HyllusEisert2006}%
  \BibitemOpen
  \bibfield  {author} {\bibinfo {author} {\bibfnamefont {P.}~\bibnamefont
  {Hyllus}}\ and\ \bibinfo {author} {\bibfnamefont {J.}~\bibnamefont
  {Eisert}},\ }\bibfield  {title} {\bibinfo {title} {Optimal entanglement
  witnesses for continuous-variable systems},\ }\href
  {https://doi.org/10.1088/1367-2630/8/4/051} {\bibfield  {journal} {\bibinfo
  {journal} {New J. Phys.}\ }\textbf {\bibinfo {volume} {8}},\ \bibinfo {pages}
  {51} (\bibinfo {year} {2006})}\BibitemShut {NoStop}%
\bibitem [{\citenamefont {Shchukin}\ and\ \citenamefont {van
  Loock}(2026)}]{ShchukinVanLoock2026}%
  \BibitemOpen
  \bibfield  {author} {\bibinfo {author} {\bibfnamefont {E.}~\bibnamefont
  {Shchukin}}\ and\ \bibinfo {author} {\bibfnamefont {P.}~\bibnamefont {van
  Loock}},\ }\bibfield  {title} {\bibinfo {title} {Revisiting gaussian genuine
  entanglement witnesses with modern software},\ }\href
  {https://doi.org/10.1103/drmr-qf2q} {\bibfield  {journal} {\bibinfo
  {journal} {Phys. Rev. Research}\ }\textbf {\bibinfo {volume} {8}},\ \bibinfo
  {pages} {023124} (\bibinfo {year} {2026})}\BibitemShut {NoStop}%
\bibitem [{\citenamefont {Gerke}\ \emph {et~al.}(2016)\citenamefont {Gerke},
  \citenamefont {Sperling}, \citenamefont {Vogel}, \citenamefont {Cai},
  \citenamefont {Roslund}, \citenamefont {Treps},\ and\ \citenamefont
  {Fabre}}]{Gerke2016}%
  \BibitemOpen
  \bibfield  {author} {\bibinfo {author} {\bibfnamefont {S.}~\bibnamefont
  {Gerke}}, \bibinfo {author} {\bibfnamefont {J.}~\bibnamefont {Sperling}},
  \bibinfo {author} {\bibfnamefont {W.}~\bibnamefont {Vogel}}, \bibinfo
  {author} {\bibfnamefont {Y.}~\bibnamefont {Cai}}, \bibinfo {author}
  {\bibfnamefont {J.}~\bibnamefont {Roslund}}, \bibinfo {author} {\bibfnamefont
  {N.}~\bibnamefont {Treps}},\ and\ \bibinfo {author} {\bibfnamefont
  {C.}~\bibnamefont {Fabre}},\ }\bibfield  {title} {\bibinfo {title}
  {Multipartite entanglement of a two-separable state},\ }\href
  {https://doi.org/10.1103/PhysRevLett.117.110502} {\bibfield  {journal}
  {\bibinfo  {journal} {Phys. Rev. Lett.}\ }\textbf {\bibinfo {volume} {117}},\
  \bibinfo {pages} {110502} (\bibinfo {year} {2016})}\BibitemShut {NoStop}%
\bibitem [{\citenamefont {Baksov{\'a}}\ \emph {et~al.}(2025)\citenamefont
  {Baksov{\'a}}, \citenamefont {Leskovjanov{\'a}}, \citenamefont {Mi{\v{s}}ta},
  \citenamefont {Agudelo},\ and\ \citenamefont {Friis}}]{Baksova2025}%
  \BibitemOpen
  \bibfield  {author} {\bibinfo {author} {\bibfnamefont {K.}~\bibnamefont
  {Baksov{\'a}}}, \bibinfo {author} {\bibfnamefont {O.}~\bibnamefont
  {Leskovjanov{\'a}}}, \bibinfo {author} {\bibfnamefont {L.}~\bibnamefont
  {Mi{\v{s}}ta}, \bibfnamefont {Jr.}}, \bibinfo {author} {\bibfnamefont
  {E.}~\bibnamefont {Agudelo}},\ and\ \bibinfo {author} {\bibfnamefont
  {N.}~\bibnamefont {Friis}},\ }\bibfield  {title} {\bibinfo {title}
  {Multi-copy activation of genuine multipartite entanglement in
  continuous-variable systems},\ }\href
  {https://doi.org/10.22331/q-2025-04-09-1699} {\bibfield  {journal} {\bibinfo
  {journal} {Quantum}\ }\textbf {\bibinfo {volume} {9}},\ \bibinfo {pages}
  {1699} (\bibinfo {year} {2025})}\BibitemShut {NoStop}%
\bibitem [{\citenamefont {Leskovjanov{\'a}}\ and\ \citenamefont
  {Mi{\v{s}}ta}(2025)}]{Leskovjanova2025}%
  \BibitemOpen
  \bibfield  {author} {\bibinfo {author} {\bibfnamefont {O.}~\bibnamefont
  {Leskovjanov{\'a}}}\ and\ \bibinfo {author} {\bibfnamefont {L.}~\bibnamefont
  {Mi{\v{s}}ta}, \bibfnamefont {Jr.}},\ }\bibfield  {title} {\bibinfo {title}
  {Minimal criteria for continuous-variable genuine multipartite
  entanglement},\ }\href {https://doi.org/10.22331/q-2025-08-27-1837}
  {\bibfield  {journal} {\bibinfo  {journal} {Quantum}\ }\textbf {\bibinfo
  {volume} {9}},\ \bibinfo {pages} {1837} (\bibinfo {year} {2025})}\BibitemShut
  {NoStop}%
\bibitem [{\citenamefont {Leskovjanov{\'a}}\ \emph {et~al.}(2026)\citenamefont
  {Leskovjanov{\'a}}, \citenamefont {Baksov{\'a}}, \citenamefont
  {Provazn{\'i}k}, \citenamefont {Mi{\v{s}}ta},\ and\ \citenamefont
  {Friis}}]{Leskovjanova2026}%
  \BibitemOpen
  \bibfield  {author} {\bibinfo {author} {\bibfnamefont {O.}~\bibnamefont
  {Leskovjanov{\'a}}}, \bibinfo {author} {\bibfnamefont {K.}~\bibnamefont
  {Baksov{\'a}}}, \bibinfo {author} {\bibfnamefont {J.}~\bibnamefont
  {Provazn{\'i}k}}, \bibinfo {author} {\bibfnamefont {L.}~\bibnamefont
  {Mi{\v{s}}ta}, \bibfnamefont {Jr.}},\ and\ \bibinfo {author} {\bibfnamefont
  {N.}~\bibnamefont {Friis}},\ }\href@noop {} {\bibinfo {title} {On the
  existence of fully inseparable biseparable gaussian states}} (\bibinfo {year}
  {2026}),\ \Eprint {https://arxiv.org/abs/2605.28404} {arXiv:2605.28404
  [quant-ph]} \BibitemShut {NoStop}%
\bibitem [{\citenamefont {Seshadreesan}\ \emph {et~al.}(2018)\citenamefont
  {Seshadreesan}, \citenamefont {Lami},\ and\ \citenamefont
  {Wilde}}]{SeshadreesanLamiWilde2018}%
  \BibitemOpen
  \bibfield  {author} {\bibinfo {author} {\bibfnamefont {K.~P.}\ \bibnamefont
  {Seshadreesan}}, \bibinfo {author} {\bibfnamefont {L.}~\bibnamefont {Lami}},\
  and\ \bibinfo {author} {\bibfnamefont {M.~M.}\ \bibnamefont {Wilde}},\
  }\bibfield  {title} {\bibinfo {title} {R{\'e}nyi relative entropies of
  quantum gaussian states},\ }\href {https://doi.org/10.1063/1.5007167}
  {\bibfield  {journal} {\bibinfo  {journal} {J. Math. Phys.}\ }\textbf
  {\bibinfo {volume} {59}},\ \bibinfo {pages} {072204} (\bibinfo {year}
  {2018})}\BibitemShut {NoStop}%
\bibitem [{\citenamefont {Werner}(1989)}]{Werner1989}%
  \BibitemOpen
  \bibfield  {author} {\bibinfo {author} {\bibfnamefont {R.~F.}\ \bibnamefont
  {Werner}},\ }\bibfield  {title} {\bibinfo {title} {Quantum states with
  einstein--podolsky--rosen correlations admitting a hidden-variable model},\
  }\href {https://doi.org/10.1103/PhysRevA.40.4277} {\bibfield  {journal}
  {\bibinfo  {journal} {Phys. Rev. A}\ }\textbf {\bibinfo {volume} {40}},\
  \bibinfo {pages} {4277} (\bibinfo {year} {1989})}\BibitemShut {NoStop}%
\bibitem [{\citenamefont {Robertson}(1929)}]{Robertson1929}%
  \BibitemOpen
  \bibfield  {author} {\bibinfo {author} {\bibfnamefont {H.~P.}\ \bibnamefont
  {Robertson}},\ }\bibfield  {title} {\bibinfo {title} {The uncertainty
  principle},\ }\href {https://doi.org/10.1103/PhysRev.34.163} {\bibfield
  {journal} {\bibinfo  {journal} {Phys. Rev.}\ }\textbf {\bibinfo {volume}
  {34}},\ \bibinfo {pages} {163} (\bibinfo {year} {1929})}\BibitemShut
  {NoStop}%
\bibitem [{\citenamefont {Schr{\"o}dinger}(1930)}]{Schrodinger1930}%
  \BibitemOpen
  \bibfield  {author} {\bibinfo {author} {\bibfnamefont {E.}~\bibnamefont
  {Schr{\"o}dinger}},\ }\bibfield  {title} {\bibinfo {title} {Zum
  heisenbergschen unsch{\"a}rfeprinzip},\ }\href@noop {} {\bibfield  {journal}
  {\bibinfo  {journal} {Sitzungsber. Preuss. Akad. Wiss., Phys.-Math. Kl.}\ ,\
  \bibinfo {pages} {296}} (\bibinfo {year} {1930})}\BibitemShut {NoStop}%
\bibitem [{\citenamefont {Douglas}(1966)}]{Douglas1966}%
  \BibitemOpen
  \bibfield  {author} {\bibinfo {author} {\bibfnamefont {R.~G.}\ \bibnamefont
  {Douglas}},\ }\bibfield  {title} {\bibinfo {title} {On majorization,
  factorization, and range inclusion of operators on hilbert space},\ }\href
  {https://doi.org/10.1090/S0002-9939-1966-0203464-1} {\bibfield  {journal}
  {\bibinfo  {journal} {Proc. Am. Math. Soc.}\ }\textbf {\bibinfo {volume}
  {17}},\ \bibinfo {pages} {413} (\bibinfo {year} {1966})}\BibitemShut
  {NoStop}%
\bibitem [{\citenamefont {Glauber}(1963)}]{Glauber1963}%
  \BibitemOpen
  \bibfield  {author} {\bibinfo {author} {\bibfnamefont {R.~J.}\ \bibnamefont
  {Glauber}},\ }\bibfield  {title} {\bibinfo {title} {Coherent and incoherent
  states of the radiation field},\ }\href
  {https://doi.org/10.1103/PhysRev.131.2766} {\bibfield  {journal} {\bibinfo
  {journal} {Phys. Rev.}\ }\textbf {\bibinfo {volume} {131}},\ \bibinfo {pages}
  {2766} (\bibinfo {year} {1963})}\BibitemShut {NoStop}%
\bibitem [{\citenamefont {Husimi}(1940)}]{Husimi1940}%
  \BibitemOpen
  \bibfield  {author} {\bibinfo {author} {\bibfnamefont {K.}~\bibnamefont
  {Husimi}},\ }\bibfield  {title} {\bibinfo {title} {Some formal properties of
  the density matrix},\ }\href {https://doi.org/10.11429/ppmsj1919.22.4_264}
  {\bibfield  {journal} {\bibinfo  {journal} {Proc. Phys.-Math. Soc. Jpn.}\
  }\textbf {\bibinfo {volume} {22}},\ \bibinfo {pages} {264} (\bibinfo {year}
  {1940})}\BibitemShut {NoStop}%
\bibitem [{\citenamefont {Chabaud}\ and\ \citenamefont
  {Mehraban}(2022)}]{ChabaudMehraban2022}%
  \BibitemOpen
  \bibfield  {author} {\bibinfo {author} {\bibfnamefont {U.}~\bibnamefont
  {Chabaud}}\ and\ \bibinfo {author} {\bibfnamefont {S.}~\bibnamefont
  {Mehraban}},\ }\bibfield  {title} {\bibinfo {title} {Holomorphic
  representation of quantum computations},\ }\href
  {https://doi.org/10.22331/q-2022-10-06-831} {\bibfield  {journal} {\bibinfo
  {journal} {Quantum}\ }\textbf {\bibinfo {volume} {6}},\ \bibinfo {pages}
  {831} (\bibinfo {year} {2022})}\BibitemShut {NoStop}%
\bibitem [{\citenamefont {Simon}\ \emph {et~al.}(1994)\citenamefont {Simon},
  \citenamefont {Mukunda},\ and\ \citenamefont
  {Dutta}}]{SimonMukundaDutta1994}%
  \BibitemOpen
  \bibfield  {author} {\bibinfo {author} {\bibfnamefont {R.}~\bibnamefont
  {Simon}}, \bibinfo {author} {\bibfnamefont {N.}~\bibnamefont {Mukunda}},\
  and\ \bibinfo {author} {\bibfnamefont {B.}~\bibnamefont {Dutta}},\ }\bibfield
   {title} {\bibinfo {title} {Quantum-noise matrix for multimode systems:
  {$U(n)$} invariance, squeezing, and normal forms},\ }\href
  {https://doi.org/10.1103/PhysRevA.49.1567} {\bibfield  {journal} {\bibinfo
  {journal} {Phys. Rev. A}\ }\textbf {\bibinfo {volume} {49}},\ \bibinfo
  {pages} {1567} (\bibinfo {year} {1994})}\BibitemShut {NoStop}%
\bibitem [{\citenamefont {Palazuelos}\ and\ \citenamefont
  {de~Vicente}(2022)}]{PalazuelosDeVicente2022}%
  \BibitemOpen
  \bibfield  {author} {\bibinfo {author} {\bibfnamefont {C.}~\bibnamefont
  {Palazuelos}}\ and\ \bibinfo {author} {\bibfnamefont {J.~I.}\ \bibnamefont
  {de~Vicente}},\ }\bibfield  {title} {\bibinfo {title} {Genuine multipartite
  entanglement of quantum states in the multiple-copy scenario},\ }\href
  {https://doi.org/10.22331/q-2022-06-13-735} {\bibfield  {journal} {\bibinfo
  {journal} {Quantum}\ }\textbf {\bibinfo {volume} {6}},\ \bibinfo {pages}
  {735} (\bibinfo {year} {2022})}\BibitemShut {NoStop}%
\bibitem [{\citenamefont {Yamasaki}\ \emph {et~al.}(2022)\citenamefont
  {Yamasaki}, \citenamefont {Morelli}, \citenamefont {Miethlinger},
  \citenamefont {Bavaresco}, \citenamefont {Friis},\ and\ \citenamefont
  {Huber}}]{Yamasaki2022}%
  \BibitemOpen
  \bibfield  {author} {\bibinfo {author} {\bibfnamefont {H.}~\bibnamefont
  {Yamasaki}}, \bibinfo {author} {\bibfnamefont {S.}~\bibnamefont {Morelli}},
  \bibinfo {author} {\bibfnamefont {M.}~\bibnamefont {Miethlinger}}, \bibinfo
  {author} {\bibfnamefont {J.}~\bibnamefont {Bavaresco}}, \bibinfo {author}
  {\bibfnamefont {N.}~\bibnamefont {Friis}},\ and\ \bibinfo {author}
  {\bibfnamefont {M.}~\bibnamefont {Huber}},\ }\bibfield  {title} {\bibinfo
  {title} {Activation of genuine multipartite entanglement: Beyond the
  single-copy paradigm of entanglement characterisation},\ }\href
  {https://doi.org/10.22331/q-2022-04-25-695} {\bibfield  {journal} {\bibinfo
  {journal} {Quantum}\ }\textbf {\bibinfo {volume} {6}},\ \bibinfo {pages}
  {695} (\bibinfo {year} {2022})}\BibitemShut {NoStop}%
\bibitem [{\citenamefont {St{\'a}rek}\ \emph {et~al.}(2026)\citenamefont
  {St{\'a}rek}, \citenamefont {Gollerthan}, \citenamefont {Leskovjanov{\'a}},
  \citenamefont {Meth}, \citenamefont {Tirler}, \citenamefont {Friis},
  \citenamefont {Ringbauer},\ and\ \citenamefont {Mi{\v{s}}ta}}]{Starek2026}%
  \BibitemOpen
  \bibfield  {author} {\bibinfo {author} {\bibfnamefont {R.}~\bibnamefont
  {St{\'a}rek}}, \bibinfo {author} {\bibfnamefont {T.}~\bibnamefont
  {Gollerthan}}, \bibinfo {author} {\bibfnamefont {O.}~\bibnamefont
  {Leskovjanov{\'a}}}, \bibinfo {author} {\bibfnamefont {M.}~\bibnamefont
  {Meth}}, \bibinfo {author} {\bibfnamefont {P.}~\bibnamefont {Tirler}},
  \bibinfo {author} {\bibfnamefont {N.}~\bibnamefont {Friis}}, \bibinfo
  {author} {\bibfnamefont {M.}~\bibnamefont {Ringbauer}},\ and\ \bibinfo
  {author} {\bibfnamefont {L.}~\bibnamefont {Mi{\v{s}}ta}, \bibfnamefont
  {Jr.}},\ }\bibfield  {title} {\bibinfo {title} {Experimental verification of
  multicopy activation of genuine multipartite entanglement},\ }\href
  {https://doi.org/10.1103/kv4s-tfc6} {\bibfield  {journal} {\bibinfo
  {journal} {Phys. Rev. Lett.}\ }\textbf {\bibinfo {volume} {136}},\ \bibinfo
  {pages} {160201} (\bibinfo {year} {2026})}\BibitemShut {NoStop}%
\bibitem [{\citenamefont {Simon}(2000)}]{Simon2000}%
  \BibitemOpen
  \bibfield  {author} {\bibinfo {author} {\bibfnamefont {R.}~\bibnamefont
  {Simon}},\ }\bibfield  {title} {\bibinfo {title} {Peres--horodecki
  separability criterion for continuous variable systems},\ }\href
  {https://doi.org/10.1103/PhysRevLett.84.2726} {\bibfield  {journal} {\bibinfo
   {journal} {Phys. Rev. Lett.}\ }\textbf {\bibinfo {volume} {84}},\ \bibinfo
  {pages} {2726} (\bibinfo {year} {2000})}\BibitemShut {NoStop}%
\bibitem [{\citenamefont {Peres}(1996)}]{Peres1996}%
  \BibitemOpen
  \bibfield  {author} {\bibinfo {author} {\bibfnamefont {A.}~\bibnamefont
  {Peres}},\ }\bibfield  {title} {\bibinfo {title} {Separability criterion for
  density matrices},\ }\href {https://doi.org/10.1103/PhysRevLett.77.1413}
  {\bibfield  {journal} {\bibinfo  {journal} {Phys. Rev. Lett.}\ }\textbf
  {\bibinfo {volume} {77}},\ \bibinfo {pages} {1413} (\bibinfo {year}
  {1996})}\BibitemShut {NoStop}%
\bibitem [{\citenamefont {von Neumann}(1931)}]{vonNeumann1931}%
  \BibitemOpen
  \bibfield  {author} {\bibinfo {author} {\bibfnamefont {J.}~\bibnamefont {von
  Neumann}},\ }\bibfield  {title} {\bibinfo {title} {Die eindeutigkeit der
  schr{\"o}dingerschen operatoren},\ }\href
  {https://doi.org/10.1007/BF01457956} {\bibfield  {journal} {\bibinfo
  {journal} {Math. Ann.}\ }\textbf {\bibinfo {volume} {104}},\ \bibinfo {pages}
  {570} (\bibinfo {year} {1931})}\BibitemShut {NoStop}%
\bibitem [{\citenamefont {Folland}(1989)}]{Folland1989}%
  \BibitemOpen
  \bibfield  {author} {\bibinfo {author} {\bibfnamefont {G.~B.}\ \bibnamefont
  {Folland}},\ }\href@noop {} {\emph {\bibinfo {title} {Harmonic Analysis in
  Phase Space}}}\ (\bibinfo  {publisher} {Princeton University Press},\
  \bibinfo {address} {Princeton},\ \bibinfo {year} {1989})\BibitemShut
  {NoStop}%
\end{thebibliography}%

\end{document}